%% file: main.tex
\documentclass[USenglish,oneside,twocolumn]{article}

\usepackage[utf8]{inputenc}
\usepackage[big]{dgruyter_NEW} 

\usepackage{natbib}
\usepackage{amsmath}
\usepackage{amsthm}
\usepackage{amsfonts}
\usepackage{mathtools}
\usepackage{amssymb}
\usepackage{microtype}
\usepackage{dirtytalk}
\usepackage{bm}
\usepackage{enumitem}
\usepackage{float}

\usepackage{xcolor, amsmath, pgfplots, tikz}
\pgfplotsset{compat=1.9}
\usepgfplotslibrary{fillbetween}
\usetikzlibrary{patterns}
\usetikzlibrary{angles, quotes}
\usetikzlibrary{calc}
\usepackage{environ}
\makeatletter
\newsavebox{\measure@tikzpicture}
\NewEnviron{scaletikzpicturetowidth}[1]{%
  \def\tikz@width{#1}%
  \def\tikzscale{1}\begin{lrbox}{\measure@tikzpicture}%
  \BODY
  \end{lrbox}%
  \pgfmathparse{#1/\wd\measure@tikzpicture}%
  \edef\tikzscale{\pgfmathresult}%
  \BODY
}
\makeatother
\usepackage{subcaption}

\DeclareMathOperator*{\argmin}{arg\,min}

\newtheorem{definition}{Definition}
\newtheorem{lemma}{Lemma}
\newtheorem{corollary}{Corollary}
\newtheorem*{theorem}{Theorem}
 
\DOI{foobar}

\cclogo{\includegraphics{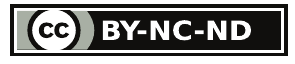}}
  
\begin{document}

\pgfmathdeclarefunction{gauss}{2}{%
        \pgfmathparse{1/(#2*sqrt(2*pi))*exp(-((x-#1)^2)/(2*#2^2))}%
    }

\author*[1]{Georgios Kaissis}
\author[2]{Moritz Knolle}
\author[3]{Friederike Jungmann}
\author[4]{Alexander Ziller}
\author[5]{Dmitrii Usynin}
\author[6]{Daniel Rueckert}

\affil[1]{Technical University of Munich, Imperial College London, g.kaissis@tum.de}
\affil[2]{Technical University of Munich, Imperial College London, moritz.knolle@tum.de}
\affil[3]{Technical University of Munich, friederike.jungmann@tum.de}
\affil[4]{Technical University of Munich, alex.ziller@tum.de}
\affil[5]{Technical University of Munich, Imperial College London, dmitrii.usynin@tum.de}
\affil[6]{Technical University of Munich, Imperial College London, 
daniel.rueckert@tum.de \newline The following authors contributed equally: Georgios Kaissis, Moritz Knolle, Friederike Jungmann}

\title{\huge A unified interpretation of the Gaussian mechanism for differential privacy through the sensitivity index}

\runningtitle{$\psi$-DP}


\begin{abstract}
{The Gaussian mechanism (GM) represents a universally employed tool for achieving differential privacy (DP), and a large body of work has been devoted to its analysis. We argue that the three prevailing interpretations of the GM, namely $(\varepsilon, \delta)$-DP, \textit{f-DP} and \textit{Rényi DP} can be expressed by using a single parameter $\psi$, which we term the \textit{sensitivity index}. $\psi$ uniquely characterises the GM and its properties by encapsulating its two fundamental quantities: the sensitivity of the query and the magnitude of the noise perturbation. With strong links to the ROC curve and the hypothesis-testing interpretation of DP, $\psi$ offers the practitioner a powerful method for interpreting, comparing and communicating the privacy guarantees of Gaussian mechanisms.}
\end{abstract}

\journalname{Submitted to PoPETs 02/2022}
\DOI{Editor to enter DOI}
\startpage{1}
\received{..}
\revised{..}
\accepted{..}
\journalyear{..}
\journalvolume{..}
\journalissue{..}

\maketitle

\section{Introduction}
Differential Privacy (DP) is the \textit{gold standard} technique for providing quantifiable privacy guarantees to individuals whose sensitive data is subjected to algorithmic processing \cite{dpbook}. Since its inception, its applicability to real-world machine learning and statistics tasks has continuously increased, as witnessed by its utilisation in the US census \cite{uscensus2018,dwork2019uscensus} and recent large-scale deployments in industry \cite{appleDP, erlingsson2014rappor}. DP is typically realised through noise perturbation of query outputs over sensitive databases.\par
The Gaussian mechanism (GM) is the prototypical mechanism for obtaining $(\varepsilon, \delta)$-DP, especially in the setting of high-dimensional queries. A large body of research has been devoted to better characterising the GM and has proposed relaxations with desirable properties, such as facilitated composition \cite{mironov2017renyi}. However, to practitioners without expert-level theoretical knowledge, translating the concepts presented in these works to tangible, real-world applications is arguably difficult, which may ultimately hinder DP's broad deployment and democratisation. Moreover, many predominantly theoretical works are also highly topical and can thus be hard to contextualise within the broader field of DP. This bears the risk of incorrectly choosing, constructing or applying DP mechanisms \cite{lyu2016understanding}, potentially resulting in the unintended disclosure of sensitive information, such as medical or financial data. Another layer of complexity emerges upon deployment of DP tools to user-facing products: For example, \textit{Apple}'s use of DP was criticised as insufficient due to its resetting of users' privacy budgets on a frequent basis and poor communication \cite{applePrivacyMacOSTang}, further aggravated by the high complexity of the subject matter.\par
Promoting trust and acceptance of DP and reinforcing its meaningful application will thus be contingent on a clear understanding of the underlying mechanisms, such as the GM. This will not only facilitate clear communication to users and developers, but also help bridge the gap between foundational DP research and its translation into task-specific tools, e.g. software libraries. Here we present a unifying perspective on the GM which encapsulates its fundamental properties into an interpretable quantity. Our contributions towards achieving the above-mentioned goal are as follows:

\begin{itemize}
    \item We demonstrate that the comprehensive characterisation of the GM requires knowledge of only a single parameter, termed the \textit{sensitivity index} $\psi$, which captures the two fundamental properties of the mechanism, namely the sensitivity of the query function and the magnitude of the noise perturbation
    \item We prove that $\psi$ fundamentally links the $(\varepsilon, \delta)$-interpretation, the hypothesis-testing interpretation of \textit{f-DP} and the divergence-based interpretation of \textit{Rényi DP} and can be used to translate between them.
    \item Due to its link to the \textit{receiver-operator characteristic} (ROC) curve, a universally applied tool in machine learning and statistics, the sensitivity index can be easily understood and communicated by and to practitioners
    \item $\psi$ leads to a intuitive geometric interpretation of the GM, as it can be used to express privacy loss in terms of the area under the ROC curve (AUC). This facilitates comparisons among mechanisms. 
    \item Finally, we theoretically and empirically demonstrate the optimal conversion strategy between $\psi$-DP and $(\varepsilon, \delta)$-DP
\end{itemize}

\subsection{Prior Work}
DP and the GM were originally introduced by Dwork et al. \cite{dpbook}. In its original form, the GM was constrained to the \textit{high privacy} regime, that is, values of $\varepsilon \in (0,1)$. This constraint was lifted in the later work by Balle and Wang \cite{balle2018improving}, who extended the analysis of the GM to arbitrary $\varepsilon$ values by employing the cumulative distribution function of the normal distribution instead of the previously used tail bounds. \textit{Rényi DP} \cite{mironov2017renyi} was proposed as a natural relaxation of DP, based on the divergence of the same name, by Mironov. It provides favourable properties under composition, however suffers from a lossy conversion to $(\varepsilon, \delta)$-DP, to which improvements have only recently been proposed \cite{balle2020hypothesis, zhu2021optimal}.\textit{f-DP} was proposed by Dong et al. \cite{dong2019gaussian} and is conceptually the closest to our work in that it expresses the properties of the DP mechanism in terms of a \textit{trade-off function}, which is similar to the ROC curve used by us. We moreover rely on the \textit{hypothesis testing interpretation} of DP and on results due to Wasserman et al \cite{wasserman2010statistical} and Kairouz et al. \cite{privacy_regions}, who introduce the notion of a \textit{privacy region}, which is conceptually related to the AUC utilised in our work. The term \textit{sensitivity index} is folklore and was first utilised as early as 1965 \cite{dempster1965expected}. It (and the synonymous term \textit{discriminability index}) is nowadays typically used to express the ratio between the mean separation and the standard deviation of two equivariant Gaussian distributions \cite{das2021method}, which motivates its utilisation in our work. We are not aware of any prior use of this term in DP literature.  

\section{Preliminaries}
We begin by introducing key terminology used in the remainder of the work. We will consider the case in which a trusted curator in possession of a sensitive database wants to employ the GM to privatise the outputs of some function applied to the database in order to privately publish the results while being able to offer the individuals whose data is contained in the database a quantifiable privacy guarantee.\par
We will refer to the sensitive database (synonymously, dataset, to designate that individuals are present in the database only once) as $D$. Its adjacent database, designated as $D'$ can be constructed from $D$ by adding or removing a single individual's data, and we will use the symbol $\wr$ to denote adjacency, i.e. that the Hamming distance between the databases equals $1$.\par
We assume that a function $f$ is applied to the database, and we denote $f$ as a \textit{query (function)}. Examples of such queries are counting how many individuals have a certain attribute. Moreover, the execution of a training step of e.g. a neural network resulting in the publication of a gradient update is also considered a query.\par
We will use a similar formalism as Laud et al. \cite{laud2020framework} to describe the global sensitivity of $f$. We consider the set $X$ of all possible databases a (\textit{Banach}) space equipped with a metric $d_X$ which expresses the distance between them. The query function $f$ then maps an element of $X$ to an element of an output space $Y$. We can now define \textit{global sensitivity} $\Delta$  of $f$ as follows:

\begin{definition}[Global sensitivity $\Delta$ of $f$]
The global sensitivity $\Delta$ of $f$ is:

\begin{equation}
    \Delta(f) = \max _{D, D' \in X} \frac{d_{Y}\left(f(D), f\left(D'\right)\right)}{d_{X}\left(D, D'\right)}
\end{equation}
When $D$ and $D'$ are adjacent as defined above and $Y$ is equipped with the $L_2$ norm, we can write:
\begin{equation}
    \Delta(f)_2 = \max_{D \wr D'} \Vert f(D) - f(D') \Vert_2
    \label{l2_sens}
\end{equation}
\end{definition}
Equation (\ref{l2_sens}) describes the global $L_2$-sensitivity of $f$, which forms the basis of our discussion. We will hence omit the subscript and refer to the global $L_2$-sensitivity as only $\Delta$.   

A DP \textit{mechanism} can be considered a higher-level function which takes as an input a database, the query function, its sensitivity and one or more hyperparameters which control how much noise to apply to the output of the query before publishing the noised result. We will limit the scope of our discussion to the \textit{Gaussian} mechanism on real-valued queries:

\begin{definition}[\textit{Gaussian} mechanism]
The Gaussian mechanism $\mathcal{M}$ on the query function $f: X \rightarrow \mathbb{R}^d$ with sensitivity $\Delta$ applied over a database $D \in X$ outputs:

\begin{equation}
    \mathcal{M}(f(D)) = f(D) + \xi, \; \xi \sim \mathcal{N}(0, \sigma^2 \, \mathbf{I}_d)
    \label{gaussian_mechanism}
\end{equation}
where $\sigma$ denotes the standard deviation of the normal distribution $\mathcal{N}$ and is calibrated to the sensitivity $\Delta$. $\mathbf{I}_d$ denotes the identity matrix with $d$ diagonal elements.
\label{def_gaussian_mech}
\end{definition}

The GM can be used to render the publication of $f$'s output differentially private with respect to the individuals in the database. 

\begin{definition}[$(\varepsilon, \delta)$-DP]
We say that the randomised mechanism $\mathcal{M}$ preserves $(\epsilon, \delta)$-DP if, for all pairs of adjacent databases $D$ and $D'$ and all subsets $\mathcal{S}$ of $\mathcal{M}$'s range:
\begin{equation}
    \mathbb{P}(\mathcal{M}(f(D) \in \mathcal{S})) \leq e^{\varepsilon} \, \mathbb{P}(\mathcal{M}(f(D') \in \mathcal{S})) + \delta
\end{equation}
We note that the definition is symmetric.
\label{eps_delta_dp_def}
\end{definition}

From the above, it becomes apparent that the publication of $\mathcal{M}$'s outputs results in \textit{privacy loss}, which affects the individuals whose records are contained in $f$'s inputs. The exact quantification of this privacy loss (for example using the parameters $(\varepsilon, \delta)$) and constraining it using specified mechanisms is central to the study of DP. \par
As we will describe in detail below, a single $(\varepsilon, \delta)$-tuple is insufficient to comprehensively describe the the privacy loss attributes of the GM. We therefore introduce two additional notions of DP aiming to better capture its properties:   

\begin{definition}[\textit{Rényi}-DP]
$M$ preserves $(\alpha, \rho)$-Rényi-DP (RDP) \cite{mironov2017renyi} if, for all pairs of adjacent databases $D$ and $D'$:
\begin{equation}
    D_{\alpha} \left( \mathcal{M}(f(D)) \parallel \mathcal{M}(f(D')) \right) \leq \rho
\end{equation}
where $D_{\alpha}$ denotes the \textit{Rényi} divergence of order $\alpha > 1$. At  $\alpha=1$, defined by continuity, this corresponds to bounding the Kullback-Leibler-divergence. Of note, $D_{\infty}$-RDP is equivalent to $(\varepsilon, 0)$-DP. The definition is asymmetric.

\end{definition}

Beyond RDP and the broader class of \textit{divergence-based} DP definitions, DP interpretations relying on the techniques of \textit{statistical hypothesis testing} have also been introduced, most notably, \textit{f-DP}:   

\begin{definition}[\textit{f-DP}]
$\mathcal{M}$ preserves f-DP if, for all pairs of adjacent databases $D$ and $D'$:
\begin{equation}
    T(\mathcal{M}(f(D)), \mathcal{M}(f(D'))) \geq f
\end{equation}
where $T$ denotes a trade-off function. 
\end{definition}

We note that, whereas RDP conveys the properties of the GM as the value of a divergence function from the outputs of $\mathcal{M}$ applied to $f(D)$ to the outputs of $\mathcal{M}$ applied to $f(D')$, \textit{f-DP} utilises the notion of a \textit{trade-off} between the \textit{Type I} and \textit{Type II} statistical errors the adversary is facing when trying to distinguish between the aforementioned outputs. Expressing the similarities between these two approaches and the $(\varepsilon, \delta)$-interpretation through the sensitivity index $\psi$ represents the core of our work. We therefore conclude this section by defining the sensitivity index.

\begin{definition}[Sensitivity index $\psi$]
The sensitivity index $\psi$ is defined as: 
\begin{equation}
    \psi = \frac{\Delta}{\sigma}
\end{equation}
\label{psi_def}
\end{definition}

As is apparent from equation (\ref{gaussian_mechanism}) and definition (\ref{def_gaussian_mech}), $\psi$ encapsulates the two parameters determining the behaviour of the GM. As the GM relies on a random noise draw and its behaviour is otherwise independent of the query function and the input database, one may state that $\psi$ is a \textit{singular} parameter of the GM.

\section{Bridging the gap between DP definitions}
Throughout this and the following sections, we will utilise the terminology and notation introduced above. We furthermore introduce an adversary $\mathcal{A}$, who has black-box access to a single output of $\mathcal{M}$ and attempts to determine whether the observed output originated from $f$ being executed over $D$ vs. $D'$. We grant the adversary access to arbitrary side-information, excluding the sensitive database itself, but including e.g. an understanding of the inner workings of $\mathcal{M}$, the characteristics of the noise distribution, the query function $f$, as well as the ability to arbitrarily post-process the output of $\mathcal{M}$. 

\subsection{Relating the sensitivity index to $(\varepsilon, \delta)$-DP}
We assume that $\mathcal{A}$ observes an output $O$ from $\mathcal{M}$. The \textit{privacy loss random variable on $O$}, $\Omega$ is then defined as:

\begin{equation}
    \Omega  = \log \left (\frac{\mathbb{P}(\mathcal{M}(f(D)) = O)}{\mathbb{P}(\mathcal{M}(f(D')) = O)} \right)
    \label{priv_loss_rv}
\end{equation}
where $\log$ is the natural logarithm. As is common in literature, we will use the notation $\Omega_{\mathcal{M}(f(D))}$ to specify that $\Omega$ quantifies the privacy loss when $\mathcal{M}$ is executed whilst considering $D$ (rather than $D'$) the \say{base} dataset \cite{dwork2016concentrated}. Under the $(\varepsilon, \delta)$-DP definition, the magnitude of the privacy loss random variable is bounded by $e^{\varepsilon}$ with probability $1-\delta$. We are therefore interested in the probability $\delta$ that the magnitude of $\Omega_{\mathcal{M}(f(D))}$ \textit{exceeds} a certain value $\varepsilon$. By Definition \ref{eps_delta_dp_def}: 

\begin{equation}
    \mathbb{P}(\Omega_{\mathcal{M}(f(D))} \geq \varepsilon) \leq \delta
\end{equation}

We begin by noting that $\Omega_{\mathcal{M}(f(D))}$ is also normally distributed \cite{dpbook}:

\begin{equation}
\Omega_{\mathcal{M}(f(D))} \sim \mathcal{N}\left( \frac{\Delta^2}{2 \sigma^2}, \frac{\Delta^2}{\sigma^2} \right) \stackrel{\text{Def.\ref{psi_def}}}{=}  \mathcal{N}\left( \frac{1}{2} \psi^2, \psi^2 \right)
\end{equation}
where we have substituted $\Vert f(D)-f(D')\Vert_2 = \Delta$, assuming the worst-case distribution without loss of generality.
The distributions of $\Omega_{\mathcal{M}(f(D))}$ and of the outputs of $\mathcal{M}$ are exemplified in Figure \ref{fig:priv_loss_distribution}. 

\begin{figure*}[]
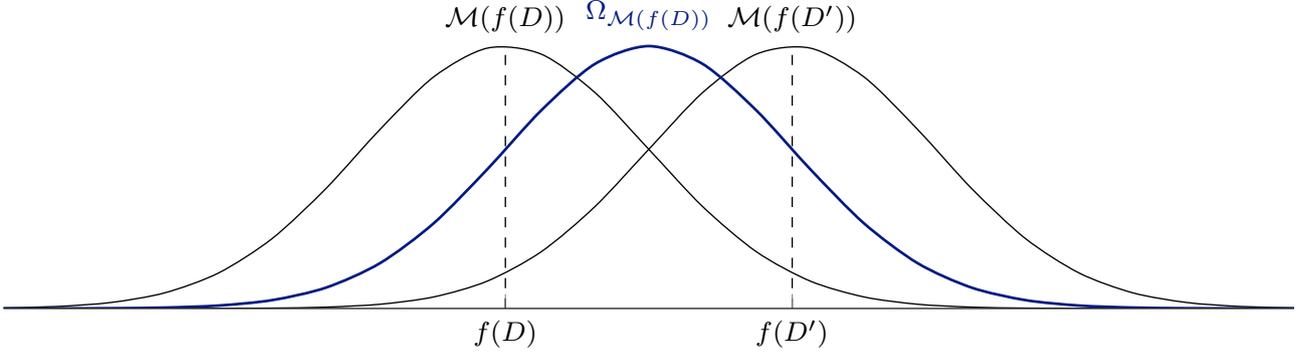

    \input Figures/figure_1.tex
    \caption{Exemplary plot showing the distributions of the Gaussian mechanism $\mathcal{M}$'s outputs given input databases $D$ and $D'$ (black curves). The \textit{privacy loss random variable} $\Omega_{\mathcal{M}(f(D))}$ (blue curve) is also normally distributed, a particular property of the Gaussian mechanism.}
    \label{fig:priv_loss_distribution}
\end{figure*}

We will bound $\delta$ using the cumulative distribution function of $\Omega_{\mathcal{M}(f(D))}$ to relate $\Omega_{\mathcal{M}(f(D))}$ to $\psi$:

\begin{lemma}
If $\mathcal{M}$ is $(\varepsilon, \delta(\varepsilon))$-DP $\forall \; \varepsilon \geq 0, \delta \in [0,1]$, the following inequality holds:
\begin{equation}
    \Phi \left ( \frac{1}{2} \psi - \frac{1}{\psi} \varepsilon \right ) \leq \delta
\end{equation}
where $\Phi$ is the cumulative distribution function of the standard normal distribution.
\label{lemma_1}
\end{lemma}

\begin{proof}
\begin{align*}
    &\mathbb{P}\left(\Omega_{\mathcal{M}(f(D)} \geq \varepsilon \right) = 
    \mathbb{P}\left(\mathcal{N}\left( \frac{1}{2} \psi^2, \psi \right)\geq \varepsilon \right) = \\ &= 
    \mathbb{P}\left(\mathcal{N}\left(0,1\right)\geq \frac{\varepsilon-\frac{1}{2} \psi^2}{\psi} \right)
\end{align*}
By symmetry of the standard normal distribution:
\begin{align*}
    & \mathbb{P}\left(\mathcal{N}\left(0,1\right)\geq \frac{\varepsilon-\frac{1}{2} \psi^2}{\psi} \right) = 
    \mathbb{P}\left(\mathcal{N}\left(0,1\right)\leq \frac{\frac{1}{2} \psi^2-\varepsilon}{\psi} \right) = \\ &=
    \Phi \left (\frac{\psi^2}{2\psi} - \frac{\varepsilon}{\psi}  \right) = 
    \Phi \left (\frac{1}{2} \psi - \frac{1}{\psi} \varepsilon \right)
\end{align*}
The inverse case, $\mathbb{P}\left(\Omega_{\mathcal{M}(f(D')} \leq -\varepsilon \right)$, required by the symmetry of Definition \ref{eps_delta_dp_def}, follows by the same argument.
\end{proof}
Notably, the $(\varepsilon, \delta(\varepsilon))$ notation in Lemma \ref{lemma_1} implies the existence of infinitely many valid $(\varepsilon, \delta)$ pairs, corresponding to the \textit{privacy profile} of $\mathcal{M}$ \cite{balle2018privacytight}. Moreover, it forms the foundation of the \textit{Analytic Gaussian Mechanism} \cite{balle2018improving}, which we will utilise later.\par
Building upon this relationship between the sensitivity index $\psi$ and the $(\varepsilon, \delta)$-DP definition, we now link $\psi$ to the hypothesis-testing DP interpretation, which we express via the ROC curve and its AUC.   

\subsection{Relating the sensitivity index to the ROC curve}
The \textit{probabilistic} definition of DP we investigated above is centred around quantifying $\mathcal{A}$'s (posterior) information gain as a function of their prior beliefs and the information disclosed by publishing the output of $\mathcal{M}$. Orthogonal to this approach, we now investigate the ability of $\mathcal{A}$ to distinguish between these outputs using a statistical hypothesis test. Formally, the adversary posits a \textit{null hypothesis} (e.g. $H_0$: \textit{the underlying database is $D$}) and an \textit{alternative hypothesis} (\textit{$H_1$: the underlying database is $D'$}) and then tries to determine whether to reject or fail to reject the null hypothesis. This interpretation is closely related to the probabilistic definition above, as the \textit{Neyman-Pearson-Lemma} \cite{neyman1933ix} motivates constructing the hypothesis test by relying on the \textit{likelihood ratio} to achieve optimal statistical power. This links it to the privacy loss random variable via the right hand side of equation (\ref{priv_loss_rv}), itself a (log-) likelihood ratio.\par
In intuitive terms, $\mathcal{A}$'s task can be summarised as follows:
\begin{enumerate}
    \item $\mathcal{A}$ formulates a binary classification problem concerned with discriminating whether the observed output $O$ of $\mathcal{M}$ is more likely given input $D$ or $D'$.
    \item $\mathcal{A}$ develops a classification algorithm $\mathcal{H}$ which takes as its input $O$ and a parameter $c$, termed the \textit{cut-off-value} or \textit{decision threshold}. $\mathcal{H}$ deterministically outputs a decision $d \in \{$\say{$D$},\say{$D'$}$\}$ depending on the observed values of $O$ and the value of $c$. 
\end{enumerate}
Our further analysis will rely on the following fundamental question: \textit{what is the optimal true-positive rate (TPR) $\mathcal{H}$ can afford $\mathcal{A}$ for any given false-positive rate (FPR) at any threshold $c$?} We will use Figure \ref{fig:tpr_fpr} to illustrate our argument.    

\begin{figure*}[h]
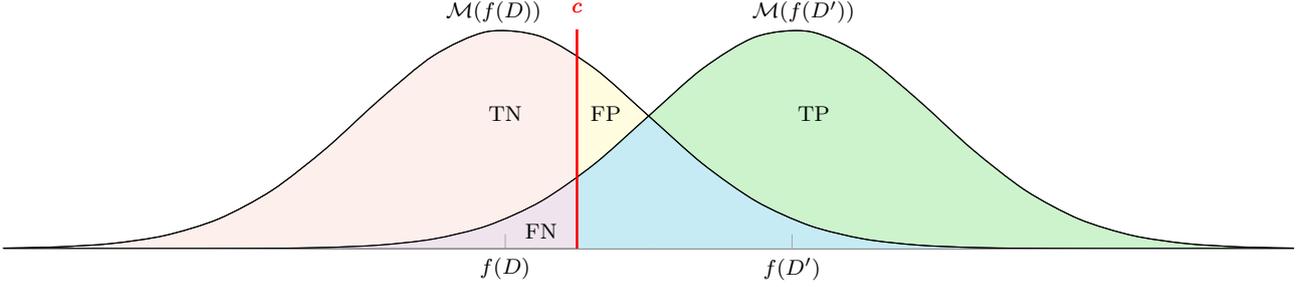

    \input Figures/figure_2.tex
    \caption{Illustration of the distributions of $\mathcal{M}$'s possible outputs and their relation to $\mathcal{A}$'s adversarial task of binary classification between $D$ and $D'$. $\mathcal{A}$'s classification algorithm $\mathcal{H}$ outputs a decision $d$ based on the \textit{cut-off} or \textit{decision threshold} $c$ (red line). At any given value of $c$, $\mathcal{H}$ achieves a classification performance which is fully described by the rates of \textit{true positive} (TP), \textit{false positive} (FP), \textit{true negative} (TN) and \textit{false negative} (FN) classifications, corresponding to the shaded areas in the figure.}
    \label{fig:tpr_fpr}
\end{figure*}

The classification performance of $\mathcal{H}$ can be fully described by its \textit{true positive rate} (TPR) and its \textit{false positive rate} (FPR) (the \textit{true negative} and \textit{false negative} rates follow by the fact that $\mbox{TNR}+\mbox{FPR} = \mbox{FNR}+\mbox{TPR} = 1$). Intuitively, if $\mathcal{A}$ cannot simultaneously achieve a high TPR \textit{and} a low FPR through the use of $\mathcal{H}$, then $\mathcal{M}$ preserves privacy. Evidently, the TPR and FPR are immediately dependent on the choice of $c$. The ROC curve represents the performance of $\mathcal{H}$ as $c$ is varied as a plot of TPR against FPR.

To create the ROC curve from the TPR and FPR, we will rely on the following property:

\begin{corollary}
Let $P_{f(D)}(x)$ be the density function of $\mathcal{M}$ when $x \sim f(D)$ and $P_{f(D')}(x)$ be the density function of $\mathcal{M}$ when $x \sim f(D')$. Then: 

\begin{align*}
    & \operatorname{TPR}(c) = \int_c^{+\infty} P_{f(D)}(t) dt = 1- \Phi_{f(D'), \sigma}(t) \\
    & \operatorname{FPR}(c) = \int^c_{-\infty} P_{f(D')}(t) dt = 1- \Phi_{f(D), \sigma}(t)
\end{align*}
where $c \in (-\infty, +\infty)$.
\end{corollary}

\begin{proof}
Notice in figure \ref{fig:tpr_fpr} that $\mathcal{M}$'s outputs follow a normal distribution. Then, $\operatorname{TPR}(c)$ is the sum of the green and blue shaded areas and $\operatorname{FPR}(c)$ is the sum of the yellow and the blue shaded areas.
\end{proof}

We are now able to relate the sensitivity index to the classification performance of $\mathcal{H}$ via the ROC curve. We will use the following facts:

\begin{corollary}[Adapted from \cite{gonccalves2014roc}, (3.1)]
Let $X$ and $Y$ be independent normal variables with means $\bm{\mu_0}$ and $\bm{\mu_1}$ and covariance matrices $\sigma_0^2\mathbf{I}$ and $\sigma_1^2\mathbf{I}$. Then, the \textit{binormal} ROC curve is given by:
\begin{equation}
    \operatorname{ROC}(t) = \Phi \left ( a + b\,\Phi^{-1}(t) \right), t \in (0,1)
    \label{roc_corollary}
\end{equation}
where $a=\frac{\Vert \bm{\mu_1}-\bm{\mu_0} \Vert_2}{\sigma}$ and $b=\frac{\sigma_0}{\sigma_1}$.
\end{corollary}

\begin{corollary}[Adapted from \cite{gonccalves2014roc}, (3.2)]
The area (AUC) under the \textit{binormal} ROC curve is then given by:
\begin{equation}
    \operatorname{AUC} = \Phi \left ( \frac{a}{\sqrt{1+b^2}} \right)
    \label{auc_corollary}
\end{equation}
where $a$ and $b$ are defined as above.
\end{corollary}

\begin{lemma}
Let $R: (0,1) \rightarrow [0,1] \times [0,1]$ be the function of the ROC curve describing the classification performance of $\mathcal{H}$. Then:

\begin{equation}
    R(x) = \Phi \left( \psi + \Phi^{-1}(x) \right)
    \label{roc_formula}
\end{equation}
\end{lemma}

\begin{proof}
The ROC curve is constructed by parametrically plotting:

\begin{equation}
(\operatorname{FPR}(c), \operatorname{TPR}(c)), \, c \in (-\infty, +\infty)
\label{roc_parametric}
\end{equation}
on the unit square and is strictly monotonically increasing \cite{gonccalves2014roc}. Recognising that $P_{f(D)}(x)$ and $P_{f(D')}(x)$ follow a normal distribution with means $f(D)$ and $f(D')$ and common standard deviation $\sigma$, we can leverage the properties of the binormal ROC curve and use equation (\ref{roc_corollary}) to re-write equation (\ref{roc_parametric}) as:

\begin{align*}
    R(x) &= \Phi \left ( \frac{\Vert f(D')-f(D) \Vert_2+\sigma \Phi^{-1}(x)}{\sigma}\right) = \\ &=
    \Phi \left ( \frac{\Delta+\sigma \Phi^{-1}(x)}{\sigma}\right)= \\ &=
    \Phi \left ( \frac{\Delta}{\sigma} +\Phi^{-1}(x)\right) = \\ &=\Phi \left( \psi + \Phi^{-1}(x) \right), x \in (0,1)
\end{align*}
\end{proof}

\begin{lemma}
The AUC of $R$ is given by:

\begin{equation}
    \operatorname{AUC}_R = \Phi \left( \frac{\psi}{\sqrt{2}} \right)
\label{auc_formula}
\end{equation}
\end{lemma}

\begin{proof}
Substitute $a=\frac{\Vert f(D)-f(D')\Vert_2}{\sigma}=\frac{\Delta}{\sigma}=\psi$ and $b=\frac{\sigma}{\sigma}=1$ in equation (\ref{auc_corollary}).
\end{proof}
Note that in both cases, we have substituted $\Vert f(D)- f(D')\Vert_2 = \Delta$, as we are interested in the \textit{worst-case} ROC curve and AUC, respectively.

From equations (\ref{roc_formula}) and (\ref{auc_formula}), the following relationships are observed:
\begin{equation}
    \lim_{\psi \rightarrow +\infty} \operatorname{AUC}_R = 1
\end{equation}
This case corresponds to a mechanism which offers \textit{no privacy}, as $\mathcal{A}$ is able to distinguish between $D$ and $D'$ with a $\operatorname{TPR}=1$ and a $\operatorname{FPR}=0$ and the vertex of the graph of $R(x)$ approaches the point $(0,1)$. This situation arises when $\Delta \rightarrow +\infty$ and/or when $\sigma \rightarrow 0$. Conversely,
\begin{equation}
    \lim_{\psi \rightarrow 0} \operatorname{AUC}_R = 0.5
\end{equation}
In this case, $\mathcal{A}$ is unable to distinguish between $D$ and $D'$ and graph of $R(x)$ approaches the graph of $\operatorname{TPR}(\operatorname{FPR})=\operatorname{FPR}$, which is the line through the origin with unit slope. This situation arises when $\Delta \rightarrow 0$ and/or when $\sigma \rightarrow +\infty$ and corresponds to \textit{perfect privacy}.

Figure \ref{fig:roc_curve} visually demonstrates the ROC curve and the AUC. As the AUC depends only on the value of $\psi$, which is sufficient to fully characterise the behaviour of the GM, the results above and the AUC can be utilised to (visually) compare the privacy guarantees between GMs. This represents an advantage over $(\varepsilon, \delta)$-DP whose characterisation and comparison would require an infinite collection of $(\varepsilon, \delta(\varepsilon))$ tuples. Moreover, the two components of these tuples represent fundamentally different quantities, further encumbering the mental model.

\begin{figure}[]
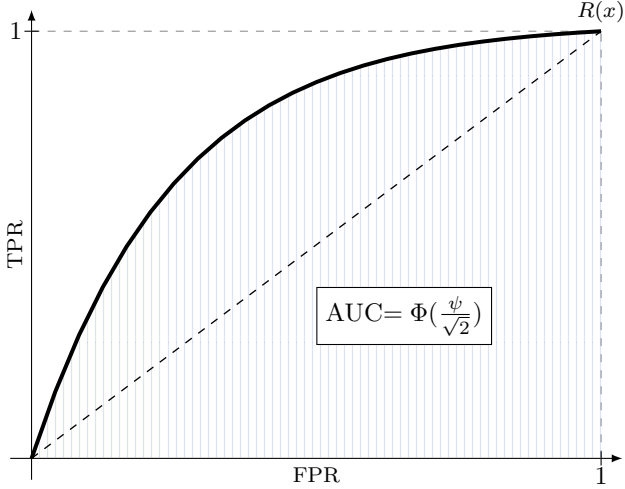

    \input Figures/figure_3.tex
    \caption{Exemplary illustration of the ROC curve $R(x)$ and of $\operatorname{AUC}_R$. The ROC curve describes the performance of a classification algorithm at all decision thresholds by plotting its TPR against its FPR on the unit square.}
    \label{fig:roc_curve}
\end{figure}

\subsection{Recovering $(\varepsilon, \delta)$-DP}
The ROC curve comprehensively characterises $\mathcal{M}$ at \textit{all} possible threshold values. We now turn to the question of how to losslessly convert between the \textit{privacy profile} of $\mathcal{M}$ and $R(x)$. From the works of \cite{wasserman2010statistical} and \cite{privacy_regions}, the following relationship is known:
\begin{corollary}[\cite{wasserman2010statistical} and \cite{privacy_regions}, Theorem (2.1)]
For a classification algorithm $\mathcal{H}$ executed on the outputs of an $(\varepsilon, \delta)$-DP mechanism, the following holds:
\begin{equation}
\operatorname{TPR}_{\mathcal{H}} \leq e^{\epsilon} \, \operatorname{FPR}_{\mathcal{H}} + \delta
\end{equation}
\end{corollary}
This implies a linear relationship between $\operatorname{TPR}_{\mathcal{H}}$ and $\operatorname{FPR}_{\mathcal{H}}$. Moreover, the \textit{optimal} classifier's graph in this setting must also be tangent to $R(x)$, which represents the worst case in terms of privacy \cite{metz1978basic}. This also follows intuitively by the fact that the vertex of such a classifier's graph lies closest to the point of perfect discrimination $(\operatorname{FPR}, \operatorname{TPR}) = (0,1)$.\par
Based on these assumptions, we can formulate the following relationship:
\begin{lemma}
The tangent with slope $e^{\varepsilon}$ to the graph of a classifier $\mathcal{H}$ with ROC curve $R(x)$ has intercept at most $\delta$.
\label{roc_to_delta}
\end{lemma}

\begin{proof}
To derive the tagent's intercept, we rely on the \textit{Legendre-Fenchel}-transform of $R$ (which, as strictly monotonically increasing, is amenable to such transformation). Let $R^{\star}$ be the convex conjugate of $R$. Then:
\begin{equation}
    R^{\star}(t) = ty - R(y)
    \label{fenchel}
\end{equation}
The condition $\frac{dR^{\star}}{dt} = 0$ is satisfied when:
\begin{equation}
    t = \Phi \left( m \right ), m = -\frac{\Delta^2+2\sigma^2\log(y)}{2\Delta\sigma}
\end{equation}
Then, substituting $t$ into equation (\ref{roc_formula}):
\begin{align*}
    R(t) &= \Phi \left(\frac{\Delta}{\sigma} + \Phi^{-1} \left( \Phi \left( m \right ) \right ) \right ) = \\ &=
    \Phi \left(\frac{\Delta}{\sigma} + m \right) = \\ &=
    \Phi \left(\frac{\Delta}{\sigma} -\frac{\Delta^2+2\sigma^2\log(y)}{2\Delta\sigma} \right) = \\ &=
    \Phi \left( \frac{\Delta}{\sigma} - \frac{\Delta^2}{2\Delta\sigma} - \frac{2\sigma^2\log(y)}{2\Delta\sigma} \right) = \\ &=
    \Phi \left( \frac{\Delta}{2\sigma} - \frac{\sigma\log(y)}{\Delta} \right)
\end{align*}
When $t=e^{\varepsilon}$:

\begin{align*}
    R(t) \big |_{t=e^{\varepsilon}} &= \Phi \left( \frac{\Delta}{2\sigma} - \frac{\sigma\log(e^{\varepsilon})}{\Delta} \right) = \\ &=
    \Phi \left( \frac{\Delta}{2\sigma} - \frac{\sigma\varepsilon}{\Delta} \right) \stackrel{\text{Lemma \ref{lemma_1}}}{=} \\ &=
    \Phi \left ( \frac{1}{2} \psi - \frac{1}{\psi} \varepsilon \right ) \leq \delta
\end{align*}
\end{proof}
This relationship is visualised in Figure \ref{fig:delta_slope}. We note that a similar case is observed when the slope of the line is $e^{-\varepsilon}$ (required by the symmetric definition of $(\varepsilon, \delta)$-DP), whereby the tangent in question is reflected about the diagonal $y=-x+1$. The resulting set of infinitely many symmetric line pairs draw the boundary of the ROC curve and are conceptually equivalent to the lines bounding the \textit{privacy region} in \cite{privacy_regions} and the \textit{trade-off} function in \cite{dong2019gaussian}.

\begin{figure}[]
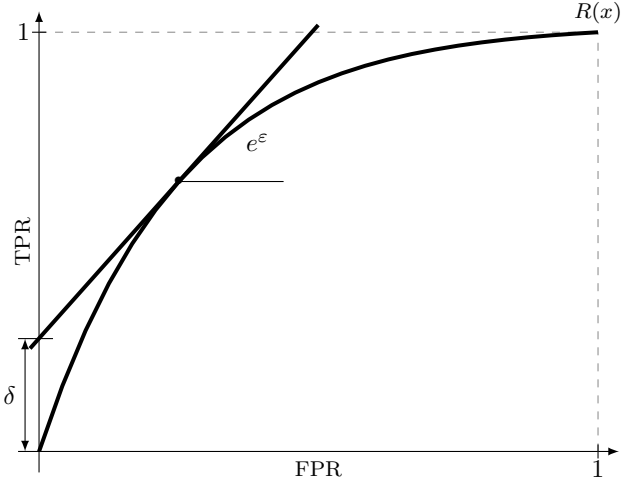

    \input Figures/figure_4.tex
    \caption{Relationship between the tangent with slope $e^{\varepsilon}$ to the ROC curve $R(x)$ and its intercept $\delta$. We note that, although we have plotted the probability $\delta$ and the TPR (which also represents a probability) on the same axis, the two are distinct.}
    \label{fig:delta_slope}
\end{figure}

Of note, this fact also admits an interesting geometric interpretation: The (fictitious) ROC curve of a \say{failed} privacy mechanism, that is, a mechanism for which the $\delta$-probability event of catastrophic failure has occurred, is shown in Figure \ref{fig:failed_mechanism}. The curve intercepts the $y$-axis at a point $(0, \kappa), \kappa >0$, indicating that $\mathcal{A}$ is able to discriminate between $D$ and $D'$ with a FPR of $0$ while still achieving a TPR of $\kappa$. A similar result can be found in Figure 3 of \cite{Pless}.

\begin{figure}[]
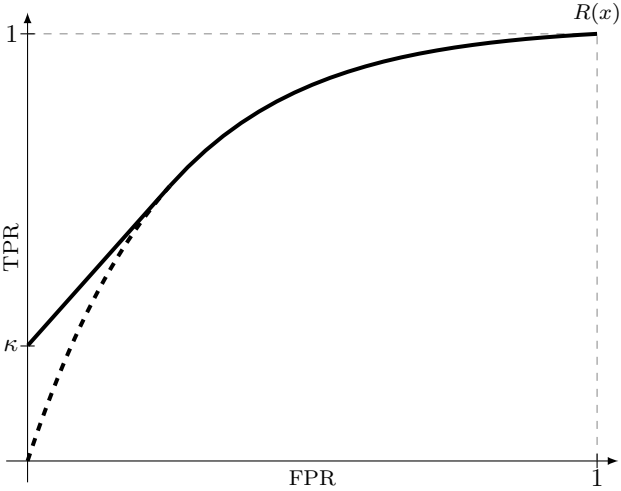

    \input Figures/figure_5.tex
    \caption{Fictitious ROC curve of a failed privacy mechanism $\mathcal{M}$ (solid line) compared to its actual ROC curve (dashed continuation). The adversary $\mathcal{A}$ is able to distinguish between $D$ and $D'$ with a TPR of $\kappa>0$ while maintaining an FPR of $0$.}
    \label{fig:failed_mechanism}
\end{figure}

\section{Relating the sensitivity index to f-DP and Rényi DP}
In the previous sections, we introduced the sensitivity index $\psi$ as a \textit{singular} value comprehensively characterising the GM and related it to the mechanism's $(\varepsilon, \delta(\varepsilon))$ privacy profile. We now study the relationship between $\psi$ and two newer DP interpretations, \textit{f-DP} and \textit{Rényi} DP.

\subsection{Converting between $\psi$-DP and f-DP}
We begin by recalling that the \textit{f-DP} framework \cite{dong2019gaussian} is based on the \textit{hypothesis testing} interpretation of DP and employs \textit{trade-off} functions between $\mathcal{A}$'s \textit{Type I} ($\alpha$) and \textit{Type II} ($\beta$) errors to characterise the behaviour of mechanisms. When the mechanism analysed is the GM, this strongly links the ROC curve used in our work to a particular subtype of \textit{f-DP}, namely \textit{Gaussian DP} (GDP). Indeed, it is easily shown that the trade-off function used in GDP is an affine transformation of a specific instance of the ROC curve:

\begin{corollary}[Equation (6) in \cite{dong2019gaussian}]
The trade-off function $G_\mu$ is given by:
\begin{equation}
G_\mu \coloneqq T\left( \mathcal{N}(0,1),\mathcal{N}(\mu,1) \right)     
\end{equation}
and
\begin{equation}
G_\mu(\alpha) = \Phi \left(\Phi^{-1}(1-\alpha) -\mu) \right)
\end{equation}
\end{corollary}

\begin{lemma}
Let $R(x)$ be the ROC curve of $\mathcal{M}$ defined above and let $\Delta=\mu$ and $\sigma=1$. Then:
\begin{equation}
    R(x) = 1-G_{\mu}(x)
\end{equation}
\end{lemma}

\begin{proof}
By substituting in equation (\ref{roc_formula}) and using the fact that $\psi=\frac{\Delta}{\sigma}=\frac{\mu}{1}=\mu$:
\begin{align*}
R(x) &= \Phi \left( \mu + \Phi^{-1}(x) \right) = 1-\left(-\Phi^{-1}(x) - \mu \right) = \\ &= 1-\left(\Phi^{-1}(1-x) - \mu \right) = 1-G_{\mu}(x)      
\end{align*}
where we have used the facts that $\Phi(x)=1-\Phi(-x)$ and $-\Phi^{-1}(x)=\Phi^{-1}(1-x)$. This operation corresponds to reflecting the graph of $R(x)$ about the $x$-axis and then transposing its $x$-coordinates by $+1$. 
\end{proof}

We will also use the following property of $G_{\mu}$:

\begin{corollary}
$G_{\mu}$ is strictly monotonically decreasing in $\mu$ such that if $\mu_1 \geq \mu_2$, then $G_{\mu_1} \leq G_{\mu_2}$.
\end{corollary}
\begin{proof}
The claim follows from the fact that $R(x)$ is strictly monotonically increasing and from the reflection operation.
\end{proof}

We can now show the following connection between \textit{f-DP} and the $\psi$-based DP interpretation:

\begin{lemma}
Let $\mathcal{M}$ be a GM with sensitivity index $\psi$ on the query function $f$ over adjacent databases $D$ and $D'$. Then, $\mathcal{M}$ is $\psi$-GDP if and only if $\psi \leq \mu$.
\label{psi_to_gdp}
\end{lemma}

\begin{proof}
The trade-off function between $\mathcal{M}_{f(D)}$ and $\mathcal{M}_{f(D')}$ is:
\begin{equation}
    T\left( \mathcal{N}(f(D),\sigma),\mathcal{N}(f(D'),\sigma) \right) = G_{\frac{\Vert f(D) - f(D') \Vert_2}{\sigma}}
\end{equation}
But: 
\begin{equation}
\frac{\Vert f(D) - f(D') \Vert_2}{\sigma} \leq \frac{\Delta}{\sigma} = \psi
\end{equation}
Hence:
\begin{equation}
T\left( \mathcal{M}_{f(D)},\mathcal{M}_{f(D')} \right) \geq G_{\psi} 
\end{equation}
As $G_{\mu}$ is strictly monotonically decreasing in $\mu$, $G_{\psi} \geq G_{\mu} \Leftrightarrow \psi \leq \mu$. But, if $G_{\psi} \geq G_{\mu}$, $\mathcal{M}$ is $\psi$-GDP. Thus, $\psi \leq \mu$ is necessary and sufficient for $\mathcal{M}$ to be $\psi$-GDP.    
\end{proof}

This relationship between \textit{f-DP} and $\psi$-based interpretations of the GM is noteworthy, as it endows them with the attractive properties of GDP, while maintaining the advantages of the (in our opinion) more intuitive ROC-curve-based GM characterisation. In particular, the following properties are a direct consequence of Lemma \ref{psi_to_gdp}:

\begin{corollary}[Group Privacy, Theorem 2.14 in \cite{dong2019gaussian}]
If $\mathcal{M}$ is $\psi$-GDP, then it is $k \psi$-GDP for groups of size $k$.
\end{corollary}

\begin{corollary}[Composition, Corollary 3.3 in \cite{dong2019gaussian}]
Let $\mathcal{M}_i$ be a sequence of $\psi_i$-GDP GMs, $i \in [1 \dots n]$. Then their n-fold composition is $\sqrt{\psi_1^2 + \dots + \psi_n^2}$-GDP.
\end{corollary}

Moreover, $\psi$-based GM interpretations are amenable to the \textit{subsampling amplification} and \textit{central-limit-theorem-type} phenomenon arising in \textit{f-DP}. This allows the privacy analysis of GMs iterated over many steps on (secret) subsamples of a database, such as DP-SGD \cite{abadi2016deep}:

\begin{corollary}[DP-SGD analysis, Corollary 5.4 in \cite{dong2019gaussian}]
Let $\mathcal{D}$ be an instance of the DP-SGD algorithm with Gaussian noise magnitude $\sigma$, executed on a database of cardinality $n$ for $T$ iterations, where $n \rightarrow +\infty$ and  $T \rightarrow +\infty$. If secret subsamples of the database are drawn uniformly at random with sampling rate $r$ at every iteration, then $\mathcal{D}$ is $\psi$-GDP with $\psi$ given by:
\begin{equation}
    \psi = s \sqrt{2} \, \sqrt{\exp \left(\frac{1}{\sigma^2} \right) \Phi \left(\frac{3}{2\sigma}\right ) + 3\Phi\left(-\frac{1}{2\sigma}\right ) -2}
\end{equation}
where $r\sqrt{T} \rightarrow s$.
\end{corollary}
Of note, this formulation is asymptotic, yet reasonable for large databases and large numbers of iterations, as are common in deep learning, and allows one to avoid cumbersome composition computations otherwise required for the analysis of subsampled, iterated GMs by \textit{f-DP}. Recently, \cite{Asoodeh2020-do} showed that for a specific range of $\sigma$ values, RDP composition can be more tight than \textit{f-DP}. Moreover, newer techniques \cite{zhu2021optimal} utilising the \textit{characteristic function} for facilitated accounting have been proposed and await investigation beyond the proof-of-principle, as well as integration into DP machine learning libraries.\par

\subsection{Relating $\psi$-DP and RDP}
RDP \cite{mironov2017renyi} was proposed as a natural relaxation of DP tailored to the specific properties of the GM. We choose to study the relationship between the sensitivity index and RDP over other divergence-based DP relaxation such as concentrated DP \cite{dwork2016concentrated} or zero-concentrated DP \cite{zcdp}, as it enjoys the greatest popularity among machine learning practitioners, representing the basis of privacy accounting of both major DP machine learning libraries (\textit{Opacus} and \textit{TensorFlow Privacy}). Fortunately, RDP can easily be expressed using the sensitivity index $\psi$, facilitating the conversion between our $\psi$-based characterisation of the GM and RDP:

\begin{lemma}
Let $\mathcal{M}$ be a GM with sensitivity index $\psi$ on the query function $f$ over adjacent databases $D$ and $D'$. Then, $\mathcal{M}$ satisfies $(\alpha, \rho)$-RDP, $\alpha \geq 1$ if and only if it also satisfies $(\alpha, \frac{\alpha}{2}\psi^2)$-RDP.
\label{rdp_lemma}
\end{lemma}

We will rely on the following fact about the \textit{Rényi} divergence from one Gaussian probability distribution to another:

\begin{corollary}
Let $P_i \coloneqq \mathcal{N}(\mu_i, \sigma_i\mathbf{I})$ and $P_j \coloneqq \mathcal{N}(\mu_j, \sigma_j\mathbf{I})$ be two Gaussian probability distributions. Then, the \textit{Rényi} divergence of order $\alpha$ from $P_i$ to $P_j$ is given by:
\begin{align*}
	& D_{\alpha}(P_i \parallel P_j) = \\ &= \log \left (\frac{\sigma_j}{\sigma_i} \right) + \frac{1}{2(\alpha-1)} 
	\log \left( \frac{\sigma_j^2}{\alpha\sigma_j^2+(1-\alpha)\sigma_i^2} \right) 
	+ \\ &+
	\frac{1}{2}\frac{\alpha(\bm{\mu}_i-\bm{\mu}_j)(\bm{\mu}_i-\bm{\mu}_j)^T}{\alpha\sigma_j^2+(1-\alpha)\sigma_i^2}
\end{align*}	
\label{renyi_divergence}
\end{corollary}

We can now prove Lemma \ref{rdp_lemma}:

\begin{proof}
We recall that the density functions of $\mathcal{M}$ over $f(D)$ and $f(D')$ follow Gaussian distributions with means $f(D)$ and $f(D')$, respectively and common covariance matrices $\sigma \mathbf{I}$. Substituting in Corollary (\ref{renyi_divergence}) yields:
\begin{align*}
    & D_{\alpha} \left( \mathcal{N}(f(D), \sigma \mathbf{I}) \parallel \mathcal{N}(f(D'), \sigma \mathbf{I}) \right) =  \log \left( \frac{\sigma}{\sigma} \right) + \\ &+ \frac{1}{2(\alpha-1)} \log \left( \frac{\sigma^2}{\alpha\sigma^2+(1-\alpha)\sigma^2} \right) + \\ & + \frac{1}{2}\frac{\alpha(f(D) - f(D'))(f(D) - f(D'))^T}{\alpha\sigma^2+(1-\alpha)\sigma^2} = \\ &=
    \frac{1}{2}\frac{\alpha(f(D) - f(D'))(f(D) - f(D'))^T}{\sigma^2} = \\ &=
    \frac{\alpha}{2} \frac{\Vert f(D) - f(D') \Vert_2^2}{\sigma^2} \leq \frac{\alpha}{2}\frac{\Delta^2}{\sigma^2} = \frac{\alpha}{2}\psi^2
\end{align*}
Thus, for $\mathcal{M}$ to satisfy $(\alpha, \rho)$-RDP it is sufficient for $\rho = \frac{\alpha}{2}\psi^2$. The inverse condition follows by the same argument.
\end{proof}

We note that this result represents a generalisation of Proposition 7 and Corollary 3 of \cite{mironov2017renyi}. Evidently, the translation between $\psi$ and $\rho$ also allows to make use of the RDP composition theorem, including its subsampled variants, for which we refer to \cite{renyi_poisson, renyi_sgm}.

\section{Optimal conversion}
In this section, we investigate the optimal conversion strategy between $\psi$-DP and $(\varepsilon, \delta)$-DP under a fixed value of $\delta$. This question bears elaboration seeing as $(\varepsilon, \delta)$-DP is \textendash in our view \textendash widely regarded as the \textit{canonical} version of DP and stakeholders may be more accustomed to it. Moreover, $\delta$ is usually viewed as a \textit{hyperparameter} whose choice depends on the size of the database in question. Our results admit two different strategies to achieve conversion: (1) Directly converting from $\psi$-DP using Lemma \ref{roc_to_delta} or (2) converting to $(\varepsilon, \delta)$-DP by taking a \say{detour} via \textit{f-DP} or RDP. We note that method (1) and the conversion via \textit{f-DP} are equivalent in the sense that the conversion between $\psi$-DP and $\psi$-GDP are immediate and rely only on a single parameter. The conversion via RDP is more involved, as it relies on the optimal choice of a second parameter ($\alpha$). Moreover, the \say{standard} conversion via RDP as presented in \cite{mironov2017renyi} is \textit{lossy}, with a tight(-er) conversion proposed in \cite{balle2020hypothesis} and (nearly identically) in \cite{Asoodeh2020-do}. We will investigate all three techniques, and begin by briefly introducing the conversion rules used in our experiments:

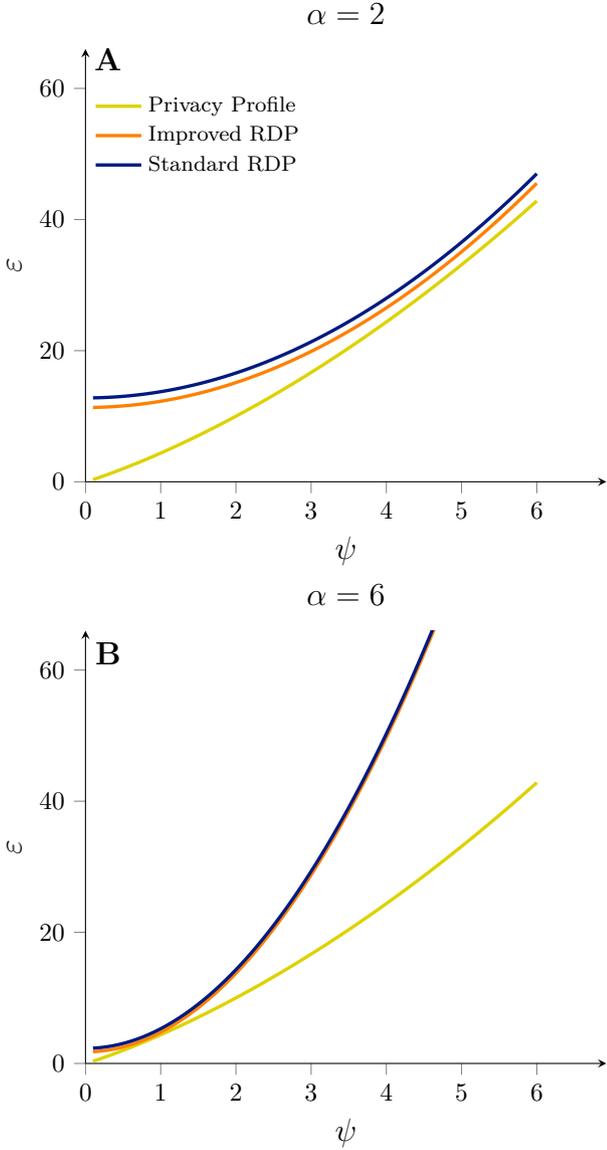
\begin{figure}[]
        \input{Figures/figure_6}\\
        \input{Figures/figure_7}
\caption{Exemplary numerical experiment results showcasing three strategies for converting between $\psi$ and $(\varepsilon, \delta(\varepsilon))$-DP. The privacy profile-based conversion (yellow curve) consistently provides the best conversion compared to the Standard RDP (blue curve) and Improved RDP (orange curve) methods. Moreover, the two RDP-based conversion techniques are $\alpha$-dependent. In the high-$\psi$ regime (\textbf{A}), RDP conversions utilising lower $\alpha$-values (here: optimal $\alpha=2$) more closely approximate the privacy profile-based conversion. In the low-$\psi$ regime (\textbf{B}), higher $\alpha$ values (here: $\alpha=6$) lead to a better approximation, which \textendash in this case \textendash almost perfectly matches the privacy profile-based conversion within a narrow interval. Observe also the increase in error when $\psi \rightarrow 0$, which is due to the different definition ranges of the RDP-based vs. the privacy profile-based conversion techniques.}
\label{numerical_experiments}
\end{figure}

\begin{corollary}[Standard RDP conversion \cite{mironov2017renyi}]
If a mechanism $\mathcal{M}$ satisfies $(\alpha, \rho)-RDP$, it also satisfies $\left(\rho +\frac{\log\left(\frac{1}{\delta}\right)}{\alpha-1}, \delta \right)$-DP for $\delta \in (0,1)$. Equivalently, $\mathcal{M}$ satisfies:
\begin{equation}
    \left(\frac{\alpha}{2}\psi^2 +\frac{\log\left(\frac{1}{\delta}\right)}{\alpha-1}, \delta \right)\mbox{-DP}
\end{equation}
\end{corollary}

\begin{corollary}[Improved RDP conversion \cite{balle2020hypothesis}]
If a mechanism $\mathcal{M}$ satisfies $(\alpha, \rho)-RDP$, it also satisfies $\left(\rho+\log\left(\frac{\alpha-1}{\alpha}\right)-\frac{\log (\delta)+\log (\alpha)}{\alpha-1},\delta \right)$-DP for $\delta \in (0,1)$. Equivalently, $\mathcal{M}$ satisfies:
\begin{equation}
    \left(\frac{\alpha}{2}\psi^2+\log\left(\frac{\alpha-1}{\alpha}\right)-\frac{\log (\delta)+\log (\alpha)}{\alpha-1},\delta \right)\mbox{-DP}
\end{equation}

\end{corollary}

Lastly, due the relationship between $\psi$ and $\delta$ proven in Lemma \ref{roc_to_delta}, we can leverage $\mathcal{M}$'s \textit{privacy profile} to convert between $\psi$-based descriptions of the GM and $(\varepsilon, \delta)$-DP. We will use the following fact from \cite{balle2018improving}, which was foreshadowed in Lemma \ref{lemma_1}:

\begin{corollary}[Analytic Gaussian Mechanism, \cite{balle2018improving}]
$\mathcal{M}$ preserves $(\varepsilon, \delta(\varepsilon))$-DP if and only if the following holds $\forall \, \varepsilon > 0, \delta \in [0,1]$:
\begin{align*}
\delta(\varepsilon) &\geq \Phi \left (\frac{\Delta}{2 \sigma} - \frac{\varepsilon \sigma}{\Delta} \right) - e^{\varepsilon} \Phi \left (-\frac{\Delta}{2 \sigma} - \frac{\varepsilon \sigma}{\Delta} \right) = \\ &=
\Phi \left (\frac{1}{2}\psi - \frac{1}{\psi}\varepsilon \right) - e^{\varepsilon} \Phi \left (-\frac{1}{2}\psi - \frac{1}{\psi}\varepsilon \right)
\end{align*}
\end{corollary}

Therefore, $\varepsilon$ given $\psi_g$ and $\delta_g$ can be found by solving:
\begin{equation}
\varepsilon_{\psi_g, \delta_g} = \argmin_x \left\{\delta(x ,\psi_g)-\delta_g \right \}
\label{numeric_epsilon}
\end{equation}
which we will refer to as the \textit{privacy profile}-based conversion.

We selected the following conditions for our numerical experiments: $\delta=10^{-5}$ (a value often used in literature, e.g. \cite{abadi2016deep}), $\psi \in [0.1, 6]$, $\alpha \in [1,64]$. For solving equation (\ref{numeric_epsilon}), we used a numerical solver employing \textit{Brent's} method \cite{brent2013algorithms}. Expectedly, using the privacy profile resulted in the best conversion, including when $\psi \rightarrow 0$. Also the improved RDP conversion method expectedly resulted in a tighter conversion compared to the standard RDP conversion. Interestingly, the two techniques performed almost on par in the higher $\alpha$ regime. Neither RDP-based technique applies when $\delta=0$ due to the logarithm and a $\delta$ term in the denominator, and thus both experience an increase in conversion error near zero.\par
These results corroborate our claim that the single-parameter characterisation provided by our technique facilitates the interpretation of the GM's behaviour: The two RDP-based methods obviously exhibit $\alpha$-dependent behaviour. For values of $\psi > 3$, the privacy profile-based direct conversion was better approximated by RDP-based conversion using low $\alpha$-values. Nevertheless, the \textit{lossiness} of the conversion was more evident in this region. Conversely, for low values of $\psi$ around $1$, the RDP-based conversion performed better at high $\alpha$-values and also matched the privacy-profile based conversion almost perfectly, albeit only within a narrow interval. These results are visualised in Figure \ref{numerical_experiments}.\par
Our experiments support that, whereas RDP-based conversions are in principle able to very closely match the privacy profile-based conversion shown in equation (\ref{numeric_epsilon}), especially in the low-$\psi$-regime, they introduce an \textendash in our view \textendash undesirable additional consideration via the $\alpha$ parameter. We therefore recommend the privacy profile-based strategy for converting between $\psi$-based DP interpretations (or, equivalently, \textit{f-DP}) and $(\varepsilon, \delta(\varepsilon))$-DP.

\section{Conclusion}
We conclude our work by summarising our findings on the properties of the Gaussian Mechanism, viewed through the lens of the sensitivity index $\psi$:
\begin{theorem}
\upshape{Let $\mathcal{M}$ be an instance of the Gaussian Mechanism with sensitivity index $\psi$. Then, all of the following statements are equivalent:}
\begin{enumerate}[label=\textbullet, font=\upshape]
    \item $\mathcal{M}$ is $(\varepsilon, \delta(\varepsilon))$-DP with:
    \begin{equation}
        \delta(\varepsilon) \geq \Phi \left (\frac{1}{2}\psi - \frac{1}{\psi}\varepsilon \right) - e^{\varepsilon} \Phi \left (-\frac{1}{2}\psi - \frac{1}{\psi}\varepsilon \right)
    \end{equation}
    \item $\mathcal{M}$ is $\psi$-GDP
    \item $\mathcal{M}$ is $\left(\alpha, \frac{\alpha}{2} \, \psi^2 \right)$-RDP
    \item The ROC curve fully characterising the properties of $\mathcal{M}$ is given by:
    \begin{equation}
        R(x)_\mathcal{M} = \Phi \left( \psi + \Phi^{-1}(x) \right)
    \end{equation}
    \item The area under $R(x)_\mathcal{M}$ is given by:
    \begin{equation}
        \operatorname{AUC}_{R_\mathcal{M}} = \Phi \left( \frac{\psi}{\sqrt{2}} \right)
    \end{equation}
\end{enumerate}
\end{theorem}

The sensitivity index provides a parsimonious, yet comprehensive description of the Gaussian Mechanism's properties. We are hopeful that future work will expand upon and improve our formalism and propose similar analyses of other DP mechanisms. We contend that the broad deployment of next-generation privacy-enhancing technologies depends \textendash among other factors \textendash on how easy DP-based tooling is to deploy, explain and interpret, and are optimistic that our work represents a worthy contribution in this direction.

\bibliography{bibliography}
\end{document}

%% file: Figures/figure_1.tex
\begin{scaletikzpicturetowidth}{\textwidth} 
    \begin{tikzpicture}[scale=\tikzscale]

        \definecolor{darkblue}{rgb}{0, 0.1, 0.5}

        \tikzstyle{every node}=[font=\small]

        \begin{axis}[
            no marks,
            smooth,
            axis x line*=bottom,
            axis y line*=left, 
            xmin=-9,
            xmax=9,
            ymin=0, 
            ymax=0.25,
            domain=-9:9,
            tick pos = left,
            axis line style={white!15!black},
            xtick={-2, 2},
            xticklabels={$f(D)$, $f(D')$},
            axis on top = true,
            ymajorticks=false,
            axis y line=none,
            height=5cm,
            width=15cm,
        ]
            \addplot[black]{gauss(-2, 2)};
            \addplot +[mark=none, dashed, black] coordinates {(-2, 0) (-2, 0.2)} node[pos=1.1]{$\mathcal{M}(f(D))$};
            \addplot[black]{gauss(2, 2)};
            \addplot +[mark=none, dashed, black] coordinates {(2, 0) (2, 0.2)} node[pos=1.1]{$\mathcal{M}(f(D'))$};
            \addplot[thick, color=darkblue]{gauss(0,2)} node[pos=0.5, yshift=9pt]{$\Omega_{\mathcal{M}(f(D))}$};

        \end{axis}
    \end{tikzpicture}
\end{scaletikzpicturetowidth}

%% file: Figures/figure_2.tex
\begin{scaletikzpicturetowidth}{\textwidth} 
    \begin{tikzpicture}[scale=\tikzscale]
        \definecolor{color0}{rgb}{0.275191,0.194905,0.496005}
        \definecolor{color1}{rgb}{0.212395,0.359683,0.55171}
        \definecolor{color2}{rgb}{0.153364,0.497,0.557724}

        \definecolor{forestgreen}{rgb}{0.13, 0.55, 0.13}
        \definecolor{tealblue}{rgb}{0.21, 0.46, 0.53}
        \definecolor{burntorange}{rgb}{0.8, 0.33, 0.0}
        \definecolor{mangotango}{rgb}{1.0, 0.51, 0.26}
        \definecolor{navyblue}{rgb}{0.0, 0.0, 0.5}
        \definecolor{mulberry}{rgb}{0.77, 0.29, 0.55}
        \definecolor{royalfuchsia}{rgb}{0.79, 0.17, 0.57}
        \definecolor{mediumseagreen}{rgb}{0.24, 0.7, 0.44}
        \definecolor{mediumspringbud}{rgb}{0.79, 0.86, 0.54}
        \definecolor{goldenpoppy}{rgb}{0.99, 0.76, 0.0}
        \definecolor{pinegreen}{rgb}{0.0, 0.47, 0.44}
        \definecolor{plum}{rgb}{0.56, 0.27, 0.52}
        \definecolor{cornflowerblue}{rgb}{0.39, 0.58, 0.93}
        \definecolor{lime}{rgb}{0.75, 1.0, 0.0}
        \definecolor{lightskyblue}{rgb}{0.53, 0.81, 0.98}
        \definecolor{limegreen}{rgb}{0.2, 0.8, 0.2}
        \definecolor{darksalmon}{rgb}{0.91, 0.59, 0.48}

        \tikzstyle{every node}=[font=\small]

        \begin{axis}[
            no marks,
            smooth,
            axis x line*=bottom,
            axis y line*=left, 
            xmin=-9,
            xmax=9,
            ymin=0, 
            ymax=0.3,
            domain=-9:9,
            tick pos = left,
            axis line style={white!15!black},
            xtick={-2, 2},
            xticklabels={$f(D)$, $f(D')$}, 
            ymajorticks=false,
            axis y line=none,
            height=5cm,
            width=15cm,
        ]
            \addplot[name path=A, black]{gauss(-2, 2)} node[above, pos=0.38] {$\mathcal{M}(f(D))$}; 
            \addplot[name path=B, black]{gauss(2, 2)} node[above, pos=0.62] {$\mathcal{M}(f(D'))$}; %
            \addplot +[name path=C,mark=none, thick, color=red] coordinates {(-1, 0) (-1, 0.2)} node[pos=1.1]{$c$}; 

            \addplot [draw=none, name path=D] {0}; 

            \addplot [color=darksalmon!30, opacity=0.5] fill between[of=A and B, soft clip={domain=-9:-1}]; 
            \addplot [fill=plum!30, draw=none, domain=-9:-1, opacity=0.5] {gauss(2,2)}\closedcycle; 
            \addplot [color=limegreen!50, draw=none, domain==-1:9, opacity=0.5] fill between[of=B and D, soft clip={domain=-1:9}]; 
            \addplot [color=yellow!30, draw=none, opacity=0.5] fill between[of=A and B, soft clip={domain=-1:0}]; 
            \addplot [color=lightskyblue!50, opacity=0.8] fill between[of=B and D, soft clip={domain=-1:0}]; 
            \addplot [color=lightskyblue!50, opacity=0.8] fill between[of=A and D, soft clip={domain=0:9}]; 

            \addplot[black]{gauss(-2, 2)}; 
            \addplot[black]{gauss(2, 2)}; 
            \addplot +[mark=none, thick, color=red] coordinates {(-1, 0) (-1, 0.2)} node[pos=1.1]{$c$}; 

            \node[label={TN}] at (700, 100) {};
            \node [label={FN}] at (750, -8) {};
            \node [label={FP}] at (840, 100) {};
            \node [label={TP}] at (1130, 100) {};
            
            \draw [gray!70, line width=0.00](1100, 0) -- (1100, 13);
        \end{axis}
    \end{tikzpicture}
\end{scaletikzpicturetowidth}

%% file: Figures/figure_3.tex
\begin{scaletikzpicturetowidth}{\columnwidth} 
    \begin{tikzpicture}[scale=\tikzscale]

        \definecolor{darkblue}{rgb}{0, 0.1, 0.5}

        \draw[-latex, name path=X] (-0.3,0) -- (8.3,0) node[pos=0.5,anchor=north]{\small FPR};
        \draw[-latex] (0,-0.3) -- (0,6.3) node[pos=0.5,anchor=south,rotate=90]{\small TPR};

        \draw [dashed,gray!90] (8,0) -- (8,6);
        \draw [dashed,gray!90] (0,6) -- (8,6);

        \draw (8,0.1) -- (8,-0.1);
        \node at (8,0)[anchor=north]{$1$};
        \draw (0.1,6) -- (-0.1,6 );
        \node at (0,6)[anchor=east]{$1$};

        \draw [domain=0:8,pattern=vertical lines, pattern color= darkblue!50,line width=0, opacity=0.5, draw=none] plot ({\x}, {6/(1-exp(-8/2))*(1-exp(-\x/2))}) -- (8,0) -- cycle; 
        \node [draw, fill=white] at (4,2) [anchor=west] {AUC$=\Phi (\frac{\psi}{\sqrt{2}})$};

        \draw [domain=0:8,variable=\x,line width=1.5, name path=R] plot ({\x}, {6/(1-exp(-8/2))*(1-exp(-\x/2))});
        \node at (8,6) [anchor=south]{\small $R(x)$};

        \draw [domain=0:8, variable=\x, line width=0.5, dashed] plot({\x}, {(6/8)*\x});

    \end{tikzpicture}
\end{scaletikzpicturetowidth}

%% file: Figures/figure_4.tex
\begin{scaletikzpicturetowidth}{\columnwidth} 
 
    \begin{tikzpicture}[scale=\tikzscale]

        \draw[-latex] (-0.3,0) -- (8.3,0) node[pos=0.5,anchor=north]{\small FPR};
        \draw[-latex] (0,-0.3) -- (0,6.3) node[pos=0.5,anchor=south,rotate=90]{\small TPR};

        \draw [dashed,gray!90] (8,0) -- (8,6);
        \draw [dashed,gray!90] (0,6) -- (8,6);

        \draw [domain=0:8,variable=\x,line width=1.5] plot ({\x}, {6/(1-exp(-8/2))*(1-exp(-\x/2))});
        \node at (8,6) [anchor=south]{\small $R(x)$};

        \draw (8,0.1) -- (8,-0.1);
        \node at (8,0)[anchor=north]{$1$};
        \draw (0.1,6) -- (-0.1,6 );
        \node at (0,6)[anchor=east]{$1$};

        \coordinate [label=center:\textbullet] (a) at ({2},{6/(1-exp(-8/2))*(1-exp(-2/2))});

        \draw[line width=1.5] (a) -- ++(48.3:3); 
        \draw[line width=1.5] (a) -- ++({48.3+180}:3.2); 

        \coordinate (b) at ({0},{6/(1-exp(-8/2))*(1-exp(-2/2))-2*1.124});

        \draw ($(b) - (0.3,0)$) -- ++(0.5,0); 

        \draw [latex-latex] (-0.2,0) -- ($(b) - (0.2,0)$) node[pos=0.5,anchor=east]{$\delta$};

        \draw (a) -- ++(1.5,0);

        \draw -- ($(a) + (1,0)$) arc (0:48.3:1) node[pos=0.7,anchor=west]{$e^\varepsilon$}; 
    
    \end{tikzpicture}
\end{scaletikzpicturetowidth}

%% file: Figures/figure_5.tex
\begin{scaletikzpicturetowidth}{\columnwidth} 
 
    \begin{tikzpicture}[scale=\tikzscale]

        \draw[-latex] (-0.3,0) -- (8.3,0) node[pos=0.5,anchor=north]{\small FPR};
        \draw[-latex] (0,-0.3) -- (0,6.3) node[pos=0.5,anchor=south,rotate=90]{\small TPR};

        \draw [dashed,gray!90] (8,0) -- (8,6);
        \draw [dashed,gray!90] (0,6) -- (8,6);

        \draw (8,0.1) -- (8,-0.1);
        \node at (8,0)[anchor=north]{$1$};
        \draw (0.1,6) -- (-0.1,6 );
        \node at (0,6)[anchor=east]{$1$};
    
        \coordinate (a) at ({2},{6/(1-exp(-8/2))*(1-exp(-2/2))});

        \draw [domain=2:8,variable=\x,line width=1.5] plot ({\x}, {6/(1-exp(-8/2))*(1-exp(-\x/2))});
        \node at (8,6) [anchor=south]{\small $R(x)$};

        \draw [domain=0:2,variable=\x,line width=1.5, dashed] plot ({\x}, {6/(1-exp(-8/2))*(1-exp(-\x/2))});

        \draw[line width=1.5] (a) -- ++({48.3+180}:3.0);
        
        \draw (-0.1,{6/(1-exp(-8/2))*(1-exp(-2/2))-2*1.124}) -- (+0.1,{6/(1-exp(-8/2))*(1-exp(-2/2))-2*1.124});
        
        \node at (0,{6/(1-exp(-8/2))*(1-exp(-2/2))-2*1.124})[anchor=east]{$\kappa$};
        

        
    \end{tikzpicture}
\end{scaletikzpicturetowidth}

%% file: Figures/figure_6.tex
\begin{scaletikzpicturetowidth}{\columnwidth} 

    \begin{tikzpicture}
    
        \definecolor{darkblue}{rgb}{0, 0.1, 0.5}
        \definecolor{cb-clay}       {RGB}{219, 209,   0}
        
        \begin{axis}[
            no marks,
            smooth,
            legend style={at={(0,0.8)},anchor=west, draw=none, font=\small},
            legend cell align={left},
            axis x line*=bottom,
            axis y line*=left, 
            axis lines = left,
            enlargelimits=upper,
            tick align=outside,
            title={\large \(\displaystyle \alpha=2\)},
            x grid style={white!80!black},
            xlabel={\large \(\displaystyle \psi\)},
            xmin=0, 
            xmax=6.295,
            ymin=0, 
            ymax=60,
            ylabel={\large \(\displaystyle \varepsilon\)},
        ]
        
        \addplot [line width=1.25pt, cb-clay]
        table [row sep=\\] {%
        0.1 0.340669364684326\\0.15959595959596 0.567027914653326\\0.219191919191919 0.802025314907453\\0.278787878787879 1.04408737334794\\0.338383838383838 1.29235297170136\\0.397979797979798 1.54627957622723\\0.457575757575758 1.80549484235608\\0.517171717171717 2.06972800832974\\0.576767676767677 2.33877376242871\\0.636363636363636 2.61247139862448\\0.695959595959596 2.89069194997471\\0.755555555555556 3.17332981880561\\0.815151515151515 3.46029709898166\\0.874747474747475 3.75151958989053\\0.934343434343434 4.04693391678837\\0.993939393939394 4.3464853993148\\1.05353535353535 4.65012644057137\\1.11313131313131 4.95781528740413\\1.17272727272727 5.26951506110774\\1.23232323232323 5.58519298887232\\1.29191919191919 5.90481978675333\\1.35151515151515 6.22836915872459\\1.41111111111111 6.55581738586949\\1.47070707070707 6.88714298639641\\1.53030303030303 7.22232643193606\\1.58989898989899 7.5613499089982\\1.64949494949495 7.90419711700314\\1.70909090909091 8.25085309618297\\1.76868686868687 8.60130408006699\\1.82828282828283 8.95553736834627\\1.88787878787879 9.31354121674458\\1.94747474747475 9.67530474116814\\2.00707070707071 10.0408178339148\\2.06666666666667 10.4100710901219\\2.12626262626263 10.7830557429519\\2.18585858585859 11.1597636062689\\2.24545454545455 11.5401870237705\\2.30505050505051 11.9243188236923\\2.36464646464646 12.3121522783683\\2.42424242424242 12.7036810680045\\2.48383838383838 13.0988992481519\\2.54343434343434 13.4978012204162\\2.6030303030303 13.9003817060272\\2.66262626262626 14.3066357219203\\2.72222222222222 14.7165585590551\\2.78181818181818 15.130145762713\\2.84141414141414 15.5473931145596\\2.9010101010101 15.9682966162826\\2.96060606060606 16.3928524746366\\3.02020202020202 16.8210570877522\\3.07979797979798 17.2529070325838\\3.13939393939394 17.6883990533734\\3.1989898989899 18.127530051043\\3.25858585858586 18.5702970734183\\3.31818181818182 19.0166973062093\\3.37777777777778 19.4667280646764\\3.43737373737374 19.9203867859199\\3.4969696969697 20.3776710217371\\3.55656565656566 20.8385784319969\\3.61616161616162 21.3031067784875\\3.67575757575758 21.7712539191965\\3.73535353535354 22.2430178029867\\3.7949494949495 22.7183964646359\\3.85454545454545 23.1973880202098\\3.91414141414141 23.6799906627423\\3.97373737373737 24.1662026581987\\4.03333333333333 24.6560223416982\\4.09292929292929 25.1494481139791\\4.15252525252525 25.6464784380831\\4.21212121212121 26.1471118362501\\4.27171717171717 26.6513468869972\\4.33131313131313 27.1591822223799\\4.39090909090909 27.6706165254121\\4.45050505050505 28.1856485276447\\4.51010101010101 28.7042770068782\\4.56969696969697 29.2265007850123\\4.62929292929293 29.752318726016\\4.68888888888889 30.2817297340125\\4.74848484848485 30.8147327514718\\4.80808080808081 31.3513267574991\\4.86767676767677 31.8915107662219\\4.92727272727273 32.4352838252602\\4.98686868686869 32.982645014279\\5.04646464646465 33.5335934436189\\5.10606060606061 34.0881282529938\\5.16565656565657 34.646248610263\\5.22525252525253 35.2079537102596\\5.28484848484848 35.7732427736806\\5.34444444444444 36.3421150460338\\5.4040404040404 36.9145697966354\\5.46363636363636 37.4906063176575\\5.52323232323232 38.0702239232218\\5.58282828282828 38.6534219485386\\5.64242424242424 39.2401997490826\\5.7020202020202 39.8305566998153\\5.76161616161616 40.4244921944337\\5.82121212121212 41.0220056446644\\5.88080808080808 41.623096479583\\5.94040404040404 42.2277641449678\\6 42.8360081026819\\};
        
        \addplot [line width=1.25pt, orange]
        table [row sep=\\] {%
        0.1 11.3412540190563\\0.15959595959596 11.3559513458596\\0.219191919191919 11.3773968616233\\0.278787878787879 11.4055905663474\\0.338383838383838 11.4405324600317\\0.397979797979798 11.4822225426763\\0.457575757575758 11.5306608142812\\0.517171717171717 11.5858472748465\\0.576767676767677 11.647781924372\\0.636363636363636 11.7164647628579\\0.695959595959596 11.7918957903041\\0.755555555555556 11.8740750067106\\0.815151515151515 11.9630024120774\\0.874747474747475 12.0586780064045\\0.934343434343434 12.1611017896919\\0.993939393939394 12.2702737619396\\1.05353535353535 12.3861939231477\\1.11313131313131 12.508862273316\\1.17272727272727 12.6382788124447\\1.23232323232323 12.7744435405337\\1.29191919191919 12.9173564575829\\1.35151515151515 13.0670175635925\\1.41111111111111 13.2234268585624\\1.47070707070707 13.3865843424926\\1.53030303030303 13.5564900153832\\1.58989898989899 13.733143877234\\1.64949494949495 13.9165459280451\\1.70909090909091 14.1066961678166\\1.76868686868687 14.3035945965484\\1.82828282828283 14.5072412142404\\1.88787878787879 14.7176360208928\\1.94747474747475 14.9347790165055\\2.00707070707071 15.1586702010785\\2.06666666666667 15.3893095746118\\2.12626262626263 15.6266971371054\\2.18585858585859 15.8708328885594\\2.24545454545455 16.1217168289736\\2.30505050505051 16.3793489583482\\2.36464646464646 16.643729276683\\2.42424242424242 16.9148577839782\\2.48383838383838 17.1927344802337\\2.54343434343434 17.4773593654495\\2.6030303030303 17.7687324396256\\2.66262626262626 18.066853702762\\2.72222222222222 18.3717231548587\\2.78181818181818 18.6833407959158\\2.84141414141414 19.0017066259331\\2.9010101010101 19.3268206449108\\2.96060606060606 19.6586828528487\\3.02020202020202 19.997293249747\\3.07979797979798 20.3426518356056\\3.13939393939394 20.6947586104245\\3.1989898989899 21.0536135742037\\3.25858585858586 21.4192167269432\\3.31818181818182 21.791568068643\\3.37777777777778 22.1706675993032\\3.43737373737374 22.5565153189236\\3.4969696969697 22.9491112275044\\3.55656565656566 23.3484553250454\\3.61616161616162 23.7545476115468\\3.67575757575758 24.1673880870085\\3.73535353535354 24.5869767514305\\3.7949494949495 25.0133136048128\\3.85454545454545 25.4463986471554\\3.91414141414141 25.8862318784584\\3.97373737373737 26.3328132987216\\4.03333333333333 26.7861429079451\\4.09292929292929 27.246220706129\\4.15252525252525 27.7130466932732\\4.21212121212121 28.1866208693777\\4.27171717171717 28.6669432344424\\4.33131313131313 29.1540137884675\\4.39090909090909 29.647832531453\\4.45050505050505 30.1483994633987\\4.51010101010101 30.6557145843047\\4.56969696969697 31.169777894171\\4.62929292929293 31.6905893929977\\4.68888888888889 32.2181490807847\\4.74848484848485 32.7524569575319\\4.80808080808081 33.2935130232395\\4.86767676767677 33.8413172779074\\4.92727272727273 34.3958697215356\\4.98686868686869 34.9571703541241\\5.04646464646465 35.5252191756729\\5.10606060606061 36.1000161861821\\5.16565656565657 36.6815613856515\\5.22525252525253 37.2698547740813\\5.28484848484848 37.8648963514713\\5.34444444444444 38.4666861178217\\5.4040404040404 39.0752240731324\\5.46363636363636 39.6905102174034\\5.52323232323232 40.3125445506347\\5.58282828282828 40.9413270728263\\5.64242424242424 41.5768577839782\\5.7020202020202 42.2191366840904\\5.76161616161616 42.868163773163\\5.82121212121212 43.5239390511958\\5.88080808080808 44.186462518189\\5.94040404040404 44.8557341741425\\6 45.5317540190563\\};
        
        \addplot [line width=1.25pt, darkblue]
        table [row sep=\\] {%
        0.1 12.8016394055225\\0.15959595959596 12.8163367323259\\0.219191919191919 12.8377822480896\\0.278787878787879 12.8659759528136\\0.338383838383838 12.9009178464979\\0.397979797979798 12.9426079291425\\0.457575757575758 12.9910462007475\\0.517171717171717 13.0462326613127\\0.576767676767677 13.1081673108383\\0.636363636363636 13.1768501493241\\0.695959595959596 13.2522811767703\\0.755555555555556 13.3344603931768\\0.815151515151515 13.4233877985436\\0.874747474747475 13.5190633928707\\0.934343434343434 13.6214871761581\\0.993939393939394 13.7306591484059\\1.05353535353535 13.8465793096139\\1.11313131313131 13.9692476597822\\1.17272727272727 14.0986641989109\\1.23232323232323 14.2348289269999\\1.29191919191919 14.3777418440492\\1.35151515151515 14.5274029500587\\1.41111111111111 14.6838122450286\\1.47070707070707 14.8469697289589\\1.53030303030303 15.0168754018494\\1.58989898989899 15.1935292637002\\1.64949494949495 15.3769313145114\\1.70909090909091 15.5670815542828\\1.76868686868687 15.7639799830146\\1.82828282828283 15.9676266007066\\1.88787878787879 16.178021407359\\1.94747474747475 16.3951644029717\\2.00707070707071 16.6190555875447\\2.06666666666667 16.849694961078\\2.12626262626263 17.0870825235717\\2.18585858585859 17.3312182750256\\2.24545454545455 17.5821022154398\\2.30505050505051 17.8397343448144\\2.36464646464646 18.1041146631493\\2.42424242424242 18.3752431704444\\2.48383838383838 18.6531198666999\\2.54343434343434 18.9377447519157\\2.6030303030303 19.2291178260918\\2.66262626262626 19.5272390892282\\2.72222222222222 19.8321085413249\\2.78181818181818 20.143726182382\\2.84141414141414 20.4620920123993\\2.9010101010101 20.787206031377\\2.96060606060606 21.1190682393149\\3.02020202020202 21.4576786362132\\3.07979797979798 21.8030372220718\\3.13939393939394 22.1551439968907\\3.1989898989899 22.5139989606699\\3.25858585858586 22.8796021134094\\3.31818181818182 23.2519534551093\\3.37777777777778 23.6310529857694\\3.43737373737374 24.0169007053898\\3.4969696969697 24.4094966139706\\3.55656565656566 24.8088407115117\\3.61616161616162 25.214932998013\\3.67575757575758 25.6277734734747\\3.73535353535354 26.0473621378967\\3.7949494949495 26.473698991279\\3.85454545454545 26.9067840336217\\3.91414141414141 27.3466172649246\\3.97373737373737 27.7931986851878\\4.03333333333333 28.2465282944114\\4.09292929292929 28.7066060925952\\4.15252525252525 29.1734320797394\\4.21212121212121 29.6470062558439\\4.27171717171717 30.1273286209087\\4.33131313131313 30.6143991749338\\4.39090909090909 31.1082179179192\\4.45050505050505 31.6087848498649\\4.51010101010101 32.1160999707709\\4.56969696969697 32.6301632806373\\4.62929292929293 33.1509747794639\\4.68888888888889 33.6785344672509\\4.74848484848485 34.2128423439981\\4.80808080808081 34.7538984097057\\4.86767676767677 35.3017026643736\\4.92727272727273 35.8562551080018\\4.98686868686869 36.4175557405903\\5.04646464646465 36.9856045621391\\5.10606060606061 37.5604015726483\\5.16565656565657 38.1419467721177\\5.22525252525253 38.7302401605475\\5.28484848484848 39.3252817379375\\5.34444444444444 39.9270715042879\\5.4040404040404 40.5356094595986\\5.46363636363636 41.1508956038696\\5.52323232323232 41.7729299371009\\5.58282828282828 42.4017124592925\\5.64242424242424 43.0372431704444\\5.7020202020202 43.6795220705567\\5.76161616161616 44.3285491596292\\5.82121212121212 44.9843244376621\\5.88080808080808 45.6468479046552\\5.94040404040404 46.3161195606087\\6 46.9921394055225\\};
        
        \legend{Privacy Profile, Improved RDP, Standard RDP};
        
        \node[label={\large\textbf{A}}] at (30, 600) {};
        
        \end{axis}
        
    \end{tikzpicture}
\end{scaletikzpicturetowidth}

%% file: Figures/figure_7.tex
\begin{scaletikzpicturetowidth}{\columnwidth} 
    \begin{tikzpicture}
    
        \definecolor{darkblue}{rgb}{0, 0.1, 0.5}
        \definecolor{cb-clay}       {RGB}{219, 209,   0}
    
    \begin{axis}[
        no marks,
            smooth,
            legend style={at={(0,0.8)},anchor=west, draw=none},
            legend cell align={left},
            axis x line*=bottom,
            axis y line*=left, 
            axis lines = left,
            enlargelimits=upper,
            tick align=outside,
            title={\large \(\displaystyle \alpha=6\)},
            x grid style={white!80!black},
            xlabel={\large \(\displaystyle \psi\)},
            xmin=0, 
            xmax=6.295,
            ymin=0, 
            ymax=60,
            ylabel={\large \(\displaystyle \varepsilon\)},
    ]
    \addplot [line width=1.25pt, cb-clay]
    table [row sep=\\] {%
    0.1 0.340669364684326\\0.15959595959596 0.567027914653326\\0.219191919191919 0.802025314907453\\0.278787878787879 1.04408737334794\\0.338383838383838 1.29235297170136\\0.397979797979798 1.54627957622723\\0.457575757575758 1.80549484235608\\0.517171717171717 2.06972800832974\\0.576767676767677 2.33877376242871\\0.636363636363636 2.61247139862448\\0.695959595959596 2.89069194997471\\0.755555555555556 3.17332981880561\\0.815151515151515 3.46029709898166\\0.874747474747475 3.75151958989053\\0.934343434343434 4.04693391678837\\0.993939393939394 4.3464853993148\\1.05353535353535 4.65012644057137\\1.11313131313131 4.95781528740413\\1.17272727272727 5.26951506110774\\1.23232323232323 5.58519298887232\\1.29191919191919 5.90481978675333\\1.35151515151515 6.22836915872459\\1.41111111111111 6.55581738586949\\1.47070707070707 6.88714298639641\\1.53030303030303 7.22232643193606\\1.58989898989899 7.5613499089982\\1.64949494949495 7.90419711700314\\1.70909090909091 8.25085309618297\\1.76868686868687 8.60130408006699\\1.82828282828283 8.95553736834627\\1.88787878787879 9.31354121674458\\1.94747474747475 9.67530474116814\\2.00707070707071 10.0408178339148\\2.06666666666667 10.4100710901219\\2.12626262626263 10.7830557429519\\2.18585858585859 11.1597636062689\\2.24545454545455 11.5401870237705\\2.30505050505051 11.9243188236923\\2.36464646464646 12.3121522783683\\2.42424242424242 12.7036810680045\\2.48383838383838 13.0988992481519\\2.54343434343434 13.4978012204162\\2.6030303030303 13.9003817060272\\2.66262626262626 14.3066357219203\\2.72222222222222 14.7165585590551\\2.78181818181818 15.130145762713\\2.84141414141414 15.5473931145596\\2.9010101010101 15.9682966162826\\2.96060606060606 16.3928524746366\\3.02020202020202 16.8210570877522\\3.07979797979798 17.2529070325838\\3.13939393939394 17.6883990533734\\3.1989898989899 18.127530051043\\3.25858585858586 18.5702970734183\\3.31818181818182 19.0166973062093\\3.37777777777778 19.4667280646764\\3.43737373737374 19.9203867859199\\3.4969696969697 20.3776710217371\\3.55656565656566 20.8385784319969\\3.61616161616162 21.3031067784875\\3.67575757575758 21.7712539191965\\3.73535353535354 22.2430178029867\\3.7949494949495 22.7183964646359\\3.85454545454545 23.1973880202098\\3.91414141414141 23.6799906627423\\3.97373737373737 24.1662026581987\\4.03333333333333 24.6560223416982\\4.09292929292929 25.1494481139791\\4.15252525252525 25.6464784380831\\4.21212121212121 26.1471118362501\\4.27171717171717 26.6513468869972\\4.33131313131313 27.1591822223799\\4.39090909090909 27.6706165254121\\4.45050505050505 28.1856485276447\\4.51010101010101 28.7042770068782\\4.56969696969697 29.2265007850123\\4.62929292929293 29.752318726016\\4.68888888888889 30.2817297340125\\4.74848484848485 30.8147327514718\\4.80808080808081 31.3513267574991\\4.86767676767677 31.8915107662219\\4.92727272727273 32.4352838252602\\4.98686868686869 32.982645014279\\5.04646464646465 33.5335934436189\\5.10606060606061 34.0881282529938\\5.16565656565657 34.646248610263\\5.22525252525253 35.2079537102596\\5.28484848484848 35.7732427736806\\5.34444444444444 36.3421150460338\\5.4040404040404 36.9145697966354\\5.46363636363636 37.4906063176575\\5.52323232323232 38.0702239232218\\5.58282828282828 38.6534219485386\\5.64242424242424 39.2401997490826\\5.7020202020202 39.8305566998153\\5.76161616161616 40.4244921944337\\5.82121212121212 41.0220056446644\\5.88080808080808 41.623096479583\\5.94040404040404 42.2277641449678\\6 42.8360081026819\\};
    
    \addplot [line width=1.25pt, orange]
    table [row sep=\\] {%
    0.1 1.79191164235448\\0.15959595959596 1.83832425331255\\0.219191919191919 1.90604693467159\\0.278787878787879 1.99507968643162\\0.338383838383838 2.10542250859262\\0.397979797979798 2.2370754011546\\0.457575757575758 2.39003836411757\\0.517171717171717 2.56431139748151\\0.576767676767677 2.75989450124643\\0.636363636363636 2.97678767541233\\0.695959595959596 3.21499091997921\\0.755555555555556 3.47450423494707\\0.815151515151515 3.75532762031591\\0.874747474747475 4.05746107608573\\0.934343434343434 4.38090460225653\\0.993939393939394 4.72565819882831\\1.05353535353535 5.09172186580107\\1.11313131313131 5.47909560317481\\1.17272727272727 5.88777941094952\\1.23232323232323 6.31777328912522\\1.29191919191919 6.76907723770189\\1.35151515151515 7.24169125667955\\1.41111111111111 7.73561534605818\\1.47070707070707 8.2508495058378\\1.53030303030303 8.78739373601839\\1.58989898989899 9.34524803659997\\1.64949494949495 9.92441240758252\\1.70909090909091 10.5248868489661\\1.76868686868687 11.1466713607506\\1.82828282828283 11.7897659429361\\1.88787878787879 12.4541705955225\\1.94747474747475 13.13988531851\\2.00707070707071 13.8469101118984\\2.06666666666667 14.5752449756878\\2.12626262626263 15.3248899098782\\2.18585858585859 16.0958449144696\\2.24545454545455 16.8881099894619\\2.30505050505051 17.7016851348552\\2.36464646464646 18.5365703506496\\2.42424242424242 19.3927656368448\\2.48383838383838 20.2702709934411\\2.54343434343434 21.1690864204384\\2.6030303030303 22.0892119178366\\2.66262626262626 23.0306474856358\\2.72222222222222 23.993393123836\\2.78181818181818 24.9774488324371\\2.84141414141414 25.9828146114393\\2.9010101010101 27.0094904608424\\2.96060606060606 28.0574763806465\\3.02020202020202 29.1267723708516\\3.07979797979798 30.2173784314576\\3.13939393939394 31.3292945624647\\3.1989898989899 32.4625207638727\\3.25858585858586 33.6170570356817\\3.31818181818182 34.7929033778917\\3.37777777777778 35.9900597905026\\3.43737373737374 37.2085262735146\\3.4969696969697 38.4483028269275\\3.55656565656566 39.7093894507414\\3.61616161616162 40.9917861449563\\3.67575757575758 42.2954929095721\\3.73535353535354 43.620509744589\\3.7949494949495 44.9668366500068\\3.85454545454545 46.3344736258256\\3.91414141414141 47.7234206720453\\3.97373737373737 49.1336777886661\\4.03333333333333 50.5652449756878\\4.09292929292929 52.0181222331105\\4.15252525252525 53.4923095609342\\4.21212121212121 54.9878069591589\\4.27171717171717 56.5046144277846\\4.33131313131313 58.0427319668112\\4.39090909090909 59.6021595762388\\4.45050505050505 61.1828972560674\\4.51010101010101 62.7849450062969\\4.56969696969697 64.4083028269275\\4.62929292929293 66.052970717959\\4.68888888888889 67.7189486793915\\4.74848484848485 69.406236711225\\4.80808080808081 71.1148348134595\\4.86767676767677 72.8447429860949\\4.92727272727273 74.5959612291313\\4.98686868686869 76.3684895425687\\5.04646464646465 78.1623279264071\\5.10606060606061 79.9774763806465\\5.16565656565657 81.8139349052868\\5.22525252525253 83.6717035003282\\5.28484848484848 85.5507821657704\\5.34444444444444 87.4511709016137\\5.4040404040404 89.372869707858\\5.46363636363636 91.3158785845033\\5.52323232323232 93.2801975315494\\5.58282828282828 95.2658265489967\\5.64242424242424 97.2727656368448\\5.7020202020202 99.301014795094\\5.76161616161616 101.350574023744\\5.82121212121212 103.421443322795\\5.88080808080808 105.513622692247\\5.94040404040404 107.6271121321\\6 109.761911642354\\};
    
    \addplot [line width=1.25pt, darkblue]
    table [row sep=\\] {%
    0.1 2.33258509299405\\0.15959595959596 2.37899770395211\\0.219191919191919 2.44672038531116\\0.278787878787879 2.53575313707118\\0.338383838383838 2.64609595923219\\0.397979797979798 2.77774885179417\\0.457575757575758 2.93071181475713\\0.517171717171717 3.10498484812107\\0.576767676767677 3.300567951886\\0.636363636363636 3.5174611260519\\0.695959595959596 3.75566437061878\\0.755555555555556 4.01517768558664\\0.815151515151515 4.29600107095548\\0.874747474747475 4.5981345267253\\0.934343434343434 4.9215780528961\\0.993939393939394 5.26633164946788\\1.05353535353535 5.63239531644063\\1.11313131313131 6.01976905381437\\1.17272727272727 6.42845286158909\\1.23232323232323 6.85844673976478\\1.29191919191919 7.30975068834146\\1.35151515151515 7.78236470731912\\1.41111111111111 8.27628879669775\\1.47070707070707 8.79152295647737\\1.53030303030303 9.32806718665796\\1.58989898989899 9.88592148723953\\1.64949494949495 10.4650858582221\\1.70909090909091 11.0655602996056\\1.76868686868687 11.6873448113901\\1.82828282828283 12.3304393935756\\1.88787878787879 12.9948440461621\\1.94747474747475 13.6805587691495\\2.00707070707071 14.387583562538\\2.06666666666667 15.1159184263274\\2.12626262626263 15.8655633605178\\2.18585858585859 16.6365183651091\\2.24545454545455 17.4287834401015\\2.30505050505051 18.2423585854948\\2.36464646464646 19.0772438012891\\2.42424242424242 19.9334390874844\\2.48383838383838 20.8109444440807\\2.54343434343434 21.7097598710779\\2.6030303030303 22.6298853684761\\2.66262626262626 23.5713209362753\\2.72222222222222 24.5340665744755\\2.78181818181818 25.5181222830767\\2.84141414141414 26.5234880620788\\2.9010101010101 27.550163911482\\2.96060606060606 28.5981498312861\\3.02020202020202 29.6674458214911\\3.07979797979798 30.7580518820972\\3.13939393939394 31.8699680131042\\3.1989898989899 33.0031942145123\\3.25858585858586 34.1577304863213\\3.31818181818182 35.3335768285312\\3.37777777777778 36.5307332411422\\3.43737373737374 37.7491997241541\\3.4969696969697 38.988976277567\\3.55656565656566 40.250062901381\\3.61616161616162 41.5324595955958\\3.67575757575758 42.8361663602117\\3.73535353535354 44.1611831952285\\3.7949494949495 45.5075101006463\\3.85454545454545 46.8751470764651\\3.91414141414141 48.2640941226849\\3.97373737373737 49.6743512393056\\4.03333333333333 51.1059184263274\\4.09292929292929 52.5587956837501\\4.15252525252525 54.0329830115738\\4.21212121212121 55.5284804097985\\4.27171717171717 57.0452878784241\\4.33131313131313 58.5834054174507\\4.39090909090909 60.1428330268783\\4.45050505050505 61.7235707067069\\4.51010101010101 63.3256184569365\\4.56969696969697 64.948976277567\\4.62929292929293 66.5936441685986\\4.68888888888889 68.2596221300311\\4.74848484848485 69.9469101618646\\4.80808080808081 71.655508264099\\4.86767676767677 73.3854164367345\\4.92727272727273 75.1366346797709\\4.98686868686869 76.9091629932083\\5.04646464646465 78.7030013770467\\5.10606060606061 80.5181498312861\\5.16565656565657 82.3546083559264\\5.22525252525253 84.2123769509677\\5.28484848484848 86.09145561641\\5.34444444444444 87.9918443522533\\5.4040404040404 89.9135431584976\\5.46363636363636 91.8565520351428\\5.52323232323232 93.820870982189\\5.58282828282828 95.8064999996362\\5.64242424242424 97.8134390874844\\5.7020202020202 99.8416882457336\\5.76161616161616 101.891247474384\\5.82121212121212 103.962116773435\\5.88080808080808 106.054296142887\\5.94040404040404 108.16778558274\\6 110.302585092994\\};
    
        
    \node[label={\large\textbf{B}}] at (30, 58) {};
    
    \end{axis}
    
    \end{tikzpicture}
\end{scaletikzpicturetowidth}

%% file: main.bbl
\begin{thebibliography}{29}
\providecommand{\natexlab}[1]{#1}
\providecommand{\url}[1]{\texttt{#1}}
\expandafter\ifx\csname urlstyle\endcsname\relax
  \providecommand{\doi}[1]{doi: #1}\else
  \providecommand{\doi}{doi: \begingroup \urlstyle{rm}\Url}\fi

\bibitem[Abadi et~al.(2016)Abadi, Chu, Goodfellow, McMahan, Mironov, Talwar,
  and Zhang]{abadi2016deep}
M.~Abadi, A.~Chu, I.~Goodfellow, H.~B. McMahan, I.~Mironov, K.~Talwar, and
  L.~Zhang.
\newblock Deep learning with differential privacy.
\newblock In \emph{Proceedings of the 2016 ACM SIGSAC conference on computer
  and communications security}, pages 308--318, 2016.

\bibitem[Abowd(2018)]{uscensus2018}
J.~M. Abowd.
\newblock The us census bureau adopts differential privacy.
\newblock In \emph{Proceedings of the 24th ACM SIGKDD International Conference
  on Knowledge Discovery \& Data Mining}, pages 2867--2867, 2018.

\bibitem[Apple(2017)]{appleDP}
Apple.
\newblock Apple differential privacy technical overview.
\newblock \url{www.apple.com/privacy/docs/Differential_Privacy_Overview.pdf},
  2017.
\newblock Accessed: 2021-13-08.

\bibitem[Asoodeh et~al.(2020)Asoodeh, Liao, Calmon, Kosut, and
  Sankar]{Asoodeh2020-do}
S.~Asoodeh, J.~Liao, F.~P. Calmon, O.~Kosut, and L.~Sankar.
\newblock Three variants of differential privacy: Lossless conversion and
  applications.
\newblock Aug. 2020.

\bibitem[Balle and Wang(2018)]{balle2018improving}
B.~Balle and Y.-X. Wang.
\newblock Improving the gaussian mechanism for differential privacy: Analytical
  calibration and optimal denoising.
\newblock In \emph{International Conference on Machine Learning}, pages
  394--403. PMLR, 2018.

\bibitem[Balle et~al.(2018)Balle, Barthe, and Gaboardi]{balle2018privacytight}
B.~Balle, G.~Barthe, and M.~Gaboardi.
\newblock Privacy amplification by subsampling: Tight analyses via couplings
  and divergences.
\newblock \emph{arXiv preprint arXiv:1807.01647}, 2018.

\bibitem[Balle et~al.(2020)Balle, Barthe, Gaboardi, Hsu, and
  Sato]{balle2020hypothesis}
B.~Balle, G.~Barthe, M.~Gaboardi, J.~Hsu, and T.~Sato.
\newblock Hypothesis testing interpretations and renyi differential privacy.
\newblock In \emph{International Conference on Artificial Intelligence and
  Statistics}, pages 2496--2506. PMLR, 2020.

\bibitem[Brent(2013)]{brent2013algorithms}
R.~P. Brent.
\newblock \emph{Algorithms for minimization without derivatives}.
\newblock Courier Corporation, 2013.

\bibitem[Bun and Steinke(2016)]{zcdp}
M.~Bun and T.~Steinke.
\newblock Concentrated differential privacy: Simplifications, extensions, and
  lower bounds.
\newblock \emph{CoRR}, abs/1605.02065, 2016.
\newblock URL \url{http://arxiv.org/abs/1605.02065}.

\bibitem[Das and Geisler(2021)]{das2021method}
A.~Das and W.~S. Geisler.
\newblock A method to integrate and classify normal distributions, 2021.

\bibitem[Dempster and Schatzoff(1965)]{dempster1965expected}
A.~P. Dempster and M.~Schatzoff.
\newblock Expected significance level as a sensitivity index for test
  statistics.
\newblock \emph{Journal of the American Statistical Association}, 60\penalty0
  (310):\penalty0 420--436, 1965.

\bibitem[Dong et~al.(2019)Dong, Roth, and Su]{dong2019gaussian}
J.~Dong, A.~Roth, and W.~J. Su.
\newblock Gaussian differential privacy, 2019.

\bibitem[Dwork(2019)]{dwork2019uscensus}
C.~Dwork.
\newblock Differential privacy and the us census.
\newblock In \emph{Proceedings of the 38th ACM SIGMOD-SIGACT-SIGAI Symposium on
  Principles of Database Systems}, pages 1--1, 2019.

\bibitem[Dwork and Roth(2014)]{dpbook}
C.~Dwork and A.~Roth.
\newblock The algorithmic foundations of differential privacy.
\newblock \emph{Foundations and Trends® in Theoretical Computer Science},
  9\penalty0 (3–4):\penalty0 211--407, 2014.
\newblock ISSN 1551-305X.
\newblock \doi{10.1561/0400000042}.
\newblock URL \url{http://dx.doi.org/10.1561/0400000042}.

\bibitem[Dwork and Rothblum(2016)]{dwork2016concentrated}
C.~Dwork and G.~N. Rothblum.
\newblock Concentrated differential privacy.
\newblock \emph{arXiv preprint arXiv:1603.01887}, 2016.

\bibitem[Erlingsson et~al.(2014)Erlingsson, Pihur, and
  Korolova]{erlingsson2014rappor}
{\'U}.~Erlingsson, V.~Pihur, and A.~Korolova.
\newblock Rappor: Randomized aggregatable privacy-preserving ordinal response.
\newblock In \emph{Proceedings of the 2014 ACM SIGSAC conference on computer
  and communications security}, pages 1054--1067, 2014.

\bibitem[Gon{\c{c}}alves et~al.(2014)Gon{\c{c}}alves, Subtil, Oliveira, and
  Bermudez]{gonccalves2014roc}
L.~Gon{\c{c}}alves, A.~Subtil, M.~R. Oliveira, and P.~d. Bermudez.
\newblock Roc curve estimation: An overview.
\newblock \emph{REVSTAT--Statistical Journal}, 12\penalty0 (1):\penalty0 1--20,
  2014.

\bibitem[Kairouz et~al.(2015)Kairouz, Oh, and Viswanath]{privacy_regions}
P.~Kairouz, S.~Oh, and P.~Viswanath.
\newblock The composition theorem for differential privacy.
\newblock In F.~Bach and D.~Blei, editors, \emph{Proceedings of the 32nd
  International Conference on Machine Learning}, volume~37 of \emph{Proceedings
  of Machine Learning Research}, pages 1376--1385, Lille, France, 07--09 Jul
  2015. PMLR.
\newblock URL \url{http://proceedings.mlr.press/v37/kairouz15.html}.

\bibitem[Laud et~al.(2020)Laud, Pankova, and Pettai]{laud2020framework}
P.~Laud, A.~Pankova, and M.~Pettai.
\newblock A framework of metrics for differential privacy from local
  sensitivity.
\newblock \emph{Proc. Priv. Enhancing Technol.}, 2020\penalty0 (2):\penalty0
  175--208, 2020.

\bibitem[Lyu et~al.(2016)Lyu, Su, and Li]{lyu2016understanding}
M.~Lyu, D.~Su, and N.~Li.
\newblock Understanding the sparse vector technique for differential privacy.
\newblock \emph{arXiv preprint arXiv:1603.01699}, 2016.

\bibitem[Metz(1978)]{metz1978basic}
C.~E. Metz.
\newblock Basic principles of roc analysis.
\newblock In \emph{Seminars in nuclear medicine}, volume~8, pages 283--298.
  Elsevier, 1978.

\bibitem[Mironov(2017)]{mironov2017renyi}
I.~Mironov.
\newblock R{\'e}nyi differential privacy.
\newblock In \emph{2017 IEEE 30th Computer Security Foundations Symposium
  (CSF)}, pages 263--275. IEEE, 2017.

\bibitem[{Mironov} et~al.(2019){Mironov}, {Talwar}, and {Zhang}]{renyi_sgm}
I.~{Mironov}, K.~{Talwar}, and L.~{Zhang}.
\newblock {R{\'e}nyi Differential Privacy of the Sampled Gaussian Mechanism}.
\newblock \emph{arXiv e-prints}, art. arXiv:1908.10530, Aug. 2019.

\bibitem[Neyman and Pearson(1933)]{neyman1933ix}
J.~Neyman and E.~S. Pearson.
\newblock Ix. on the problem of the most efficient tests of statistical
  hypotheses.
\newblock \emph{Philosophical Transactions of the Royal Society of London.
  Series A, Containing Papers of a Mathematical or Physical Character},
  231\penalty0 (694-706):\penalty0 289--337, 1933.
\newblock \doi{10.1098/rsta.1933.0009}.
\newblock URL \url{https://doi.org/10.1098/rsta.1933.0009}.

\bibitem[Pless et~al.(2003)Pless, Larson, Siebers, and Westover]{Pless}
R.~Pless, J.~Larson, S.~Siebers, and B.~Westover.
\newblock Evaluation of local models of dynamic backgrounds.
\newblock In \emph{2003 {IEEE} Computer Society Conference on Computer Vision
  and Pattern Recognition, 2003. Proceedings.} {IEEE} Comput. Soc, 2003.
\newblock \doi{10.1109/cvpr.2003.1211454}.
\newblock URL \url{https://doi.org/10.1109/cvpr.2003.1211454}.

\bibitem[Tang et~al.(2017)Tang, Korolova, Bai, Wang, and
  Wang]{applePrivacyMacOSTang}
J.~Tang, A.~Korolova, X.~Bai, X.~Wang, and X.~Wang.
\newblock Privacy loss in apple's implementation of differential privacy on
  macos 10.12.
\newblock \emph{arXiv preprint arXiv:1709.02753}, 2017.

\bibitem[Wasserman and Zhou(2010)]{wasserman2010statistical}
L.~Wasserman and S.~Zhou.
\newblock A statistical framework for differential privacy.
\newblock \emph{Journal of the American Statistical Association}, 105\penalty0
  (489):\penalty0 375--389, 2010.

\bibitem[Zhu and Wang(2019)]{renyi_poisson}
Y.~Zhu and Y.-X. Wang.
\newblock Poission subsampled r{\'e}nyi differential privacy.
\newblock In \emph{International Conference on Machine Learning}, pages
  7634--7642. PMLR, 2019.

\bibitem[Zhu et~al.(2021)Zhu, Dong, and Wang]{zhu2021optimal}
Y.~Zhu, J.~Dong, and Y.-X. Wang.
\newblock Optimal accounting of differential privacy via characteristic
  function.
\newblock \emph{arXiv preprint arXiv:2106.08567}, 2021.

\end{thebibliography}
